\documentclass[sn-mathphys,Numbered]{sn-jnl}% Math and Physical Sciences Reference Style
\usepackage{graphicx}%
\usepackage{multirow}%
\usepackage{amsmath,amssymb,amsfonts}%
\usepackage{amsthm}%
\usepackage{mathrsfs}%
\usepackage[title]{appendix}%
\usepackage{xcolor}%
\usepackage{textcomp}%
\usepackage{manyfoot}%
\usepackage{booktabs}%
\usepackage{algorithm}%
\usepackage{algorithmicx}%
\usepackage{algpseudocode}%
\usepackage{listings}%
\theoremstyle{thmstyleone}%
\newtheorem{theorem}{Theorem}%  meant for continuous numbers
\newtheorem{proposition}{Proposition}% 

\theoremstyle{thmstyletwo}%
\newtheorem{example}{Example}%

\newtheorem{lemma}{Lemma}%

\theoremstyle{thmstylethree}%
\begin{document}

\title[Article Title]{A Stackelberg Game based on the Secretary Problem: Optimal Response is History Dependent}


\author*[1]{\fnm{David M.} \sur{Ramsey}}\email{david.ramsey@pwr.edu.pl}
\equalcont{ORCID number: 0000-0002-7186-1436}


\affil*[1]{\orgdiv{Department of Computer Science and Systems Engineering}, \orgname{Wroc\l aw University of Science and Technology}, \orgaddress{\street{Wybrze\.{z}e Wyspia\'{n}skiego 27}, \city{Wroc\l aw}, \postcode{50-370}, \country{Poland}}}

%\affil[2]{\orgdiv{Department}, \orgname{Organization}, \orgaddress{\street{Street}, %\city{City}, \postcode{10587}, \state{State}, \country{Country}}}


\abstract{This article considers a problem arising from a two-player game based on the classical secretary problem. First, Player 1 selects one object from a sequence as in the secretary problem. All of the other objects are then presented to Player 2 in the same order as in the original sequence. The goal of both players is to select the best object. The optimal response of Player 2 is adapted to the optimal strategy in the secretary problem. This means that when Player 2 observes an object that is the best seen so far, it can be inferred whether Player 1 selected one of the earlier objects in the original sequence. However, Player 2 cannot compare the current object with the one selected by Player 1. Hence, this game defines an auxiliary problem in which Player 2 has incomplete information on the relative rank of an object. It is shown that the optimal strategy of Player 2 is based on both the number of objects to have appeared and the probability that the 
current object is better than the object chosen by Player 1 (if Player 1 chose an earlier object in the sequence). However, this probability is dependent on the previously observed objects. 
A lower bound on the optimal expected reward in the auxiliary problem is defined by limiting the memory of Player 2. An upper bound is derived by giving Player 2 additional information at appropriate times. The methods used illustrate approaches that can be used to approximate the optimal reward in a stopping problem when there is incomplete information on the ranks of objects and/or the optimal strategy is history dependent, as in the Robbins' problem.}




\maketitle

\section{Introduction} \label{sec1}

 In the classical secretary problem (see Gilbert and Mosteller 1966) a decision maker (DM) 
observes a sequence of $N$ objects, which appear in random order (i.e. each permutation of the $N$ objects is assumed to be equally likely) and whose values can be linearly ranked relative to previously observed objects. It is assumed that $N$ is known. The DM can accept one object and only at the time it appears. The goal of the DM is to choose the most valuable object.  The optimal strategy is of the following form: automatically reject the first $n^{\ast}-1$ objects and then accept the next object which is relatively best, i.e. better than previously observed objects (as long as such an object appears). The threshold $n^{\ast}$ is the smallest integer $n$ such that the expected reward from accepting a relatively best object at moment $n$ is at least as great as the expected reward from continued search.  The 
asymptotic probability of choosing the best object is $e^{-1}$ and $n^{\ast}$ tends to $Ne^{-1}$. 

For the purposes of this paper, we will use the following characterization of a secretary problem as given by Ferguson (1989). There are $N$ applicants for a position (where $N$ is known). These applications appear sequentially in a random order. The applicants can be ranked according to a linear order. A DM can only employ one applicant and the decision to employ an applicant can only be made on the appearance of the candidate and on the basis of the rank of the applicant with respect to the applicants who have already appeared. The goal of a DM is to obtain the best of all the applicants and thus one can assume that a DM's 
payoff is equal to one when he/she employs the best applicant, otherwise his/her payoff is equal to zero. Bruss (2000) embeds such problems in a class of problems where the DM 
aims to stop a sequence on the occurrence of the last "success". According to this interpretation, a success at a given moment corresponds to the appearance of an object that is relatively best (i.e. the best seen so far). According to such an interpretation, the absolutely best object from a sequence of $N$ is the last relatively best object to appear. These results were generalized by Grau Ribas (2020), who assumed that the payoff from 
correctly predicting the last success could depend on the moment at which the DM stops a process. Bay\'{o}n \textit{et al.} (2023) consider a general approach to finding asymptotic solutions to similar problems by deriving and solving an appropriate differential equation.

In other variants of best choice problems, the information on the set of applicants may differ from a simple linear ordering.
In the full information best choice
problem, the DM observes the value of each object when it appears. 
When these values come from a continuous distribution, it is clear that a DM can rank the 
objects that have appeared according to a linear ordering and thus a DM has a richer source of information. 
Given that a single DM knows the distribution of the values of offers, then the asymptotic probability of choosing the most valuable object under the 
optimal policy is approximately 0.580164, independently of this distribution (see Samuels 1991; Gnedin 1996). Gnedin and Miretskiy (2007) give an explicit formula for this optimal value. Gnedin \textit{et al.} (2023) generalize this problem to one in which the 
distribution of the values of the objects is not necessarily continuous (with the condition that when the number 
of objects tends to infinity, then the mass of probability at any discrete point tends to zero).
Skarupski (2019) considers a similar problem in which the DM can pay for information on whether a current candidate is the best of all objects. Kubicka \textit{et al.} (2023) consider a problem in which the DM only observes the relative rank of each object, but the goal is to accept an object whose value is above the $(1-1/N)$-quantile of the distribution of values. In this case, the optimal strategy is to accept an object at moment $i$ if and only if its relative rank is below a threshold that is a non-decreasing function of $i$. Skarupski (2020) considers a problem in which there can be errors in observing the relative rank of an object. It is assumed that the observed relative rank of the current object is never greater than its actual rank. The optimal strategy in such a problem is to accept the first object that 
has an observed relative rank of one (a candidate) after a given time threshold. In such a problem, the DM generally starts accepting candidates at a later time than in the classical secretary problem, but has a smaller probability of accepting the absolutely best object. 

 In other variants of best choice problems, the graphical 
representation of the ordering of all the offers may differ from a linear ordering. In this case, some objects may be incomparable and/or there may be multiple best objects (see e.g. 
Garrod and Morris 2013). 
Bearden (2006) considers a similar problem that is not a best choice problem, but the goal of the DM is to maximize the expected value of the object accepted while only accepting an offer that is the best seen so far. For large $N$, the optimal 
strategy in this problem is to automatically reject the first $\sqrt{N}$ of the offers and then 
accept the next offer that is the best that has been seen so far. Ferenstein and 
Krasnosielska (2010) consider a similar problem in which objects appear according to
a Poisson process.

Secretary problems form a subset of best choice problems in which the goal of a DM is to obtain the absolutely best applicant. A significant amount of work has been directed at 
problems in which the number of applicants is random. Presman and Sonin (1973) consider a best choice problem in which the number of applicants comes from a known distribution. 
The form of the optimal strategy in such problems is not necessarily given by a time threshold. For example, suppose that the number of applicants is either 2 or $N$, where $N$ is a large integer, each with probability 0.5. In this case, if the second object is better than the first object, it is optimal to accept it, since it is likely that no more objects will appear. 
However, if a third object appears, then the DM knows that there will be in total $N$ 
objects and it is optimal to automatically reject such an object. One particular example that 
they consider assumes that the number of objects has a binomial distribution with parameters $N$ and $p$. For large $N$, under the optimal 
strategy the DM should reject approximately the first $e^{-1}Np$ offers and then accept the 
next offer that is relatively the best. As $N\rightarrow \infty$, the optimal probability of accepting the best of all offers tends to $e^{-1}$, i.e. in the limit as $N\rightarrow \infty$, this problem is equivalent to the secretary problem in which the number of offers is equal to the expected number of offers, $Np$. These results were extended by Petruccelli (1983).
Cowan and Zabczyk (1979) adapt the classical secretary problem by assuming that objects appear according to a Poisson process of known intensity $\lambda$ that occurs over the time interval $[0,T]$. The optimal strategy of the DM is of the form: accept the first relatively best object to appear after some time $t^{\ast}$. As $T\rightarrow \infty$, $t^{\ast}\rightarrow e^{-1}T$ and the 
probability of accepting the best object overall tends to $e^{-1}$. Hence, in the 
limit the Cowan-Zabczyk problem is also essentially the same as the classical secretary problem with a large number of applicants. Bruss (1984, 1987), as well as Bruss and Samuels (1987, 1990), generalize the results described immediately above to describe the robustness 
of the $e^{-1}$ threshold rule to a wide range of scenarios. R\"{u}schendorf (2016) 
derives asymptotic solutions to a range of optimal stopping problems by assuming that 
objects appear according to a Poisson process. Porosi\'{n}ski (1987) presents a model of 
the full information best choice problem when the number of objects to appear comes from a random distribution. Chudjakow and Riedel (2013) consider generalisations of the secretary problem in which, for example, the order in which candidates appear is not random.

A similar class of problems (to ones in which a random number of applicants appear) is formed by 
problems in which the number of applicants is known, but an applicant may refuse an offer of
employment from the DM. Smith (1975) considers a variant of the secretary problem in which each applicant would 
accept an offer with probability $p$, independently of the overall rank of the applicant. 
If an applicant refuses an offer, then the DM can continue searching. The goal of the 
DM is to employ the best of all the applicants. Tamaki (1991) argues that a more 
reasonable goal of the DM would be to employ the best of all the available applicants. 
He considered two models. According to one of these models, the DM can ascertain the 
availability of an applicant before a decision is made. According to the other model, the 
availability of an applicant can only be ascertaned by making an offer of employment.
Ano \textit{et al.} (1996) extends these results by assuming that the DM can make a maximum number of offers. Tamaki (2009) considers the full information version of such a 
problem, in which the DM observes the value of each successive applicant from a known 
distribution.  Ramsey (2016) considers two analogous problems in which each of $N$ applicants is 
available with probability $p$. However, according to these models, when an applicant is 
unavailable, then the DM cannot compare the value of this applicant with the value of any 
other applicant. According to one model, the goal of the DM is to employ the best of all the 
applicants. According to the other model, the goal of the DM is to, more realistically, 
employ the best of all the available candidates. The optimal strategies in these problems are 
of a similar form. Any decision on whether to accept a relatively best object depends not only 
on the number of objects that are yet to appear, but also on the number of objects that have so far turned out to be available (or, equivalently, unavailable). As $N\rightarrow \infty$, this 
dependency on the number (proportion) of previous objects that turned out to be available 
disappears. Hence, the decision on whether to accept a relatively best object only depends on the proportion of objects to have appeared so far. In fact, this optimal strategy tends to the asymptotically optimal strategy in the secretary problem, i.e. accept the first relatively 
best object to appear after a proportion $e^{-1}$ of the objects have already appeared. The 
probability of obtaining the best of the available objects is $e^{-1}$. The probability of 
obtaining the best of all the objects is $pe^{-1}$. For finite $N$, the optimal probability of obtaining 
the best of all the available objects is bounded below by the optimal probability of obtaining the best object when the number of objects comes from the binomial distribution with 
parameters $N$ and $p$ (see Presman and Sonin 1973). This is due to the fact that 
one can interpret these two models as essentially assuming that there are $N$ applicants placed in sequence and each is available with probabilty $p$. According to Ramsey's model, when an applicant is unavailable, then the 
DM can note this fact. However, according to the model of Presman and Sonin, when an 
applicant is unavailable, the DM cannot note this fact.

For the purposes of this paper, a strategy will be said to be 
history independent when the decision on whether to accept an object depends 
only on its rank relative to the objects that have already appeared, the value of such an 
object (if this value is observed) and the number/proportion of objects that have already appeared. Strategies that 
do not satisfy this condition are called history dependent. In the case of a secretary problem, a strategy is history dependent if and only if the decision on whether to accept 
a relatively best object only depends on the number/proportion of objects that have already appeared. Time threshold strategies in secretary problems form a particular class of history independent strategies under which a 
relatively best object is accepted if and only if a specified number of objects have 
already appeared (in the limiting case as $N\rightarrow \infty$, when a specified proportion of objects have already appeared).  The solution to a sequential decision problem is called 
asymptotically history independent if there exists an asymptotically optimal solution (as 
$N\rightarrow \infty$) that is history independent. 
For example, as described above, the optimal solutions to the problems considered by Ramsey (2016) are history dependent. However, asymptotically, these solutions are history independent.

Samuel-Cahn (1995) considers an adaptation of the secretary problem in which there are a fixed number of applicants, but there is a freeze on hiring after a random number $M$ of applicants have appeared. The full information version of this problem, in which the 
DM observes the values of applicants, is considered in Samuel-Cahn (1996).

Best choice problems belong to a class of problems in which the payoff of a DM depends on the rank of the object chosen. In the so called postdoc problem, the goal of the DM is to 
select the second best object from a sequence of $N$. This problem was independently solved by Szajowski (1982) and Vanderbei (1983). Liu and Milenkovic (2022) consider a similar problem in which the permutation corresponding to a sequence of the $N$ objects is not chosen from a uniform distribution. Bay\'{o}n \textit{et al.} (2018) show similarities between this problem and the problem of accepting either the best of the worst object. In both cases the first 50\% of objects should be automatically rejected.  Both Szajowski (2009) and Goldenshluger \textit{et al.} (2020) consider a class of problems in which a DM's payoff depends on the rank of the object selected.

Chow \textit{et al.} (1964) consider the problem of 
minimizing the expected rank of the object chosen when objects can be assessed according 
to a linear ranking. They show that as $N\rightarrow \infty$ the expected rank of the object selected under the optimal policy tends to $w$, where $w\approx 3.8695$.
In the Robbins' problem, the DM has the same goal, but observes the values of objects from a known distribution. This problem has been studied for a very long time, since it is technically very difficult (for the background to the problem and a number of results, see Bruss 2005). At moment $n$, the 
DM should accept the current object if its expected absolute rank is smaller than the expected rank of 
the object obtained from continued search using the optimal strategy. The expected 
absolute rank of the current object only depends on its value and relative rank compared to previous 
observations. However, the expected rank of an object obtained from continued search 
depends on the values of all the objects seen so far. Hence, the Robbins' problem can be said to be fully history dependent (see Bruss 2005). Bruss and Ferguson (1993) consider the threshold rule based on the value of the current object that maximizes the value of the object chosen (see Moser 1956). Asymptotically, the 
expected rank of the object chosen under this strategy is 7/3. Although Moser's rule works relatively well, Bruss and Ferguson show that the optimal expected reward under such a memoryless strategy is approximately 
2.3266. They show that the optimal expected rank under any strategy is asymptotically bounded below by 1.908. 
Using a similar approach, Assaf and Samuel-Cahn (1996) 
show that the optimal expected value of the rank under strategies of this form lies in the interval $(2.295,2.327)$.  Bruss and Ferguson (1996) prove that the optimal strategy is history dependent. This is done by considering the optimal strategy of a so called 
"half prophet". At any given decision point, such a half prophet has the same information as 
the DM in the standard Robbins' problem. However, given that the half prophet rejects the current object, at the next decision point he can see the entire future of the process.
Gnedin (2007) considers a generalized form of the Robbins' problem in which the cost incurred by the decision maker is assumed to be an increasing function of the rank of the object taken and show that the asymptotic form of the optimal strategy is history dependent.
Bruss and Swan (2009) consider a variant of Robbins' problem in which objects appear as a 
Poisson process at rate 1. They show that the asymptotic value function in this problem corresponds to the asymptotic value of the discrete time problem and present a differential equation that this value function must satisfy. 

More recently, Dendievel and Swan (2016) give an explicit solution of Robbins' problem for $n=4$. Allaart and Allen (2019) consider a variant of the Robbins' problem in which the observations conform to a random walk, rather than a set of independent and identically distributed random variables. By considering a strategy that depends both on the relative rank and value of the current candidate, Meier and S\"{o}gner (2017) improve the upper bound on 
the asymptotically optimal expected rank to approximately 2.32614.

The problem considered here can be interpreted as a two-player version of the secretary problem in which one of the players always has priority in deciding whether to accept or reject an object. The players jointly observe a sequence of offers and each player can obtain at most one object. The first such model was presented by Dynkin (1969), who considered a two-player game in which DM1 and DM2 could accept objects that appeared at odd and even moments, respectively. Later game-theoretical models of the secretary problem can essentially be split into two streams. In the first stream (offline choice), players choose a time threshold at the beginning of the game. Each player accepts the first object of relative rank 1 to appear after his/her time threshold. If only one player wishes to accept an object, then he/she obtains that object. However, if the players wish to accept the same object, then this object is assigned to one of the players at random and the game stops. Sakaguchi (1980) considers such a model in which the goal of each player is to obtain the most valuable object. One interesting aspect of such a model is that at equilibrium players try to avoid choosing the same object by choosing significantly different thresholds. Fushimi (1981) considers a similar model and also considers a variant in which a player who does not obtain an accepted object is free to choose any object that appears at a later time.  In the second stream (online choice), players choose whether to accept or reject an object whenever one appears. Once one player has accepted an object, the other is free to continue searching. Enns and Ferenstein (1987) consider a model in which two players simultaneously observe a sequence of objects. If both players wish to accept the same object, then Player 1 has priority. Whenever a player accepts an object, then the other is free to accept an object that appears later. Szajowski (1994) extends this model to one in which whenever both players wish to accept an object, then priority is assigned to one of the players with a given probability. In such games it is often the case that a player would prefer the other player to accept an object so that he/she could carry on searching alone, but in the worst scenario neither player accepts that object. In such scenarios, players should consider randomized strategies (see 
Neumann \textit{et al.} 2002). However, in such games it would be generally better when players coordinated their strategies rather than independently randomizing (see Ramsey and Szajowski 2008).  Szajowski (2007) considers a similar problem in which objects appear according to a Poisson process.  A more general framework for such games in which the players do not know the probability with which they have priority is considered 
by Szajowski and Skarupski (2019). Fabien \textit{et al.} (2023) consider a model in which the players can at any moment recall one of the previous candidates to appear. Bei and Zhang (2022) consider a game in which the players do not observe the candidates in the same order. They derive equilibria for both the online and offline versions of such a game. 

The game considered in this article is most similar to the following three models of sequential
games. Ramsey (2007) considers a problem in which Player 1 always has priority and when he/she accepts (and thus obtains) an object, then Player 2 does not observe that object. This leads to an auxiliary optimization problem in which the optimal response of Player 2 should be adapted to the incomplete information he/she has and the strategy of Player 1. Jacobovic (2022) gives a number of results for a similar model. Skarupski and Szajowski (2023) present a model in which the player with priority observes the values of the objects, while the other player only observes the relative ranks. Based on the decisions of the first player, at equilirium the second player obtains some information about the value of each object that has appeared.

The following two models are also in many ways similar to the one considered here. Kuchta and Morayne (2014) considered an extension of the secretary problem in which one DM can observe a set of $n$ sequences, each with $N$ objects. It is assumed that all the objects in the $i$-th sequence are observed in random order before the objects in the $(i+1)$-th sequence can be observed. The DM may only accept one object and only at the time it appears. The goal of the 
DM is to choose an object that is the best in its sequence. Intuitively, the larger the number of sequences, the more likely it is that the DM achieves his/her goal. The optimal strategy can be defined by a set of non-increasing threshold times $(t_1 ,t_2, \ldots ,t_n)$. When observing the $i$-th sequence, the DM automatically rejects the first $t_i$ objects to appear and then accepts the next relatively best object in the $i$-th sequence or passes on to the next sequence if no such object appears. When $n=1$, the problem reduces to the classical secretary problem. Asymptotically, for large $N$, the threshold
$t_i$ (understand as the proportion of offers in the $i$-th sequence that are automatically rejected) is the probability of eventually picking an object that is the best in its sequence when the DM starts observing the $i$-th sequence. Kuchta (2017) considers the corresponding extension of the full information best choice problem.

This article considers a problem which is in some ways a mirror image of the problem solved by Kuchta and Morayne (2014). Assume that $n$ DMs observe a single sequence of $N$ offers. The goal of each of the DMs is to choose the best of all the objects. The first DM
 (DM1) observes the sequence and can choose an object only when it appears. Hence, the problem 
faced by the DM1 is the classical secretary problem. The second DM (DM2) then observes the sequence with the objects appearing in the same order as observed by the DM1, except for the object chosen by the DM1. Later DMs face an analogous problem in which the 
objects chosen by the previous DMs have been removed from the initial sequence. Intuitively, the probability that the $(i+1)$-th DM chooses the best of all the objects is smaller than
the probability that the $i$-th DM chooses the best of all the objects. 

Such a problem can be interpreted as an $n$-player game in which the players are subject to a linear hierarchical structure. A Nash equilibrium of such a game can be derived by solving a sequence of nested optimization problems. DM1 should use the optimal strategy in the secretary problem. Given that DM1 uses such a strategy, DM2 then faces an auxiliary optimization problem. The solution to this problem gives the optimal reaction of DM2 to the optimal strategy of DM1. In games with a larger number of players, the auxiliary optimization problem faced by the $i$-th DM can be defined iteratively by first deriving the optimal strategies of the first $i-1$ DMs in the appropriately defined auxiliary optimization problems. 

In this article we derive the form of the optimal 
response of DM2 to the optimal strategy of DM1 and several results regarding the optimal probability with which DM2 chooses the best of all the objects. For convenience, DM1 and DM2 will be referred to as she and he, respectively. There are a number of issues regarding how one may interpret the rules of this game without essentially changing the problem faced by DM2. Firstly, it is assumed 
that before DM2 begins observing the sequence, he is not explicitly given any information regarding whether DM1 accepted an object or not. Note that DM2 should automatically reject any object that is not the best he has seen so far. Any object that is the best seen so far by DM2 will be referred to as a candidate according to DM2. Any moment at which such a candidate appears will be referred to as a decision point for DM2. Suppose that DM1 uses her optimal strategy, which is to accept the 
first relatively best object to appear no earlier than at moment $n^{\ast}$. Given that the $n$-th object observed by DM2 is a candidate, where $n<n^{\ast}$, then DM2 knows that DM1 has not yet accepted an object. Since he can compare this candidate with all of the previous objects in the original sequence, then he knows that 
such an object is the best of the first $n$ objects in the original sequence. 

Now suppose that 
  the $n$-th object observed by DM2 is a candidate, where $n\geq n^{\ast}$. In this case,  DM2 knows that this cannot be the first relatively best object to appear 
no earlier than at moment $n^{\ast}$ in the original sequence, since otherwise DM1 
would have taken this object. Hence, such an object must be the $(n+1)$-th object in the original sequence. Hence, at any decision point DM2 knows 
how many objects are yet to appear. It follows from this argument that the problem faced by DM2 is essentially the same in the following two scenarios: a) first DM1 observes the sequence alone and makes her choice, then DM2 observes the objects not chosen by DM1 in the same order, b) at moment $i$, $i=1,2,\ldots ,N$, DM1 observes the $i$-th object and then if she does not accept it, then DM2 observes the object and makes his decision. 

It will be shown that the decision on whether a candidate should be accepted by Player 2 
depends on both the probability that the current object is the best of those to appear so far and the number of objects to have appeared so far. However, when $n>n^{\ast}$ the probability that the current object is the best to have appeared so far depends on the times at which candidates appeared between moment $n^{\ast}$ and $n$. Hence, the optimal 
strategy is history dependent. It will be argued that this history dependence does not disappear asymptotically. A lower limit for the equilibrium reward of Player 2 is obtained by considering strategies that depend on the number of candidates to have appeared between 
moment $n^{\ast}$ and the current moment. By increasing the number of candidates that can be remembered, we obtain successivley larger lower bounds for the expected reward of Player 2. However, the limit of this set of bounds is strictly below his optimal expected reward. An upper limit on this expected reward is obtained by considering a problem in which 
Player 2 is given information about the relative rank of a candidate whenever such a candidate is rejected. These approaches are similar to those used to obtained bounds on the optimal expected rank of the object chosen in Robbins' problem.

Section \ref{sec2} gives a description of the classical secretary problem. It also illustrates the general approach to solving the auxiliary problem faced by Player 2 on the basis of 
two players observing a sequence of length 4.
Section \ref{sec3} then describes the form of the optimal response of Player 2 for finite $N$. The exact solution of the game when $N=50$ is also derived. 
Since the optimal strategy for $N=50$ is somewhat complex, in Section \ref{sec4} a near optimal strategy based on the number of candidates to appear after moment $n^{\ast}$ is derived.  A number of asymptotic results are given in Section \ref{sec5}. Upper bounds are derived using the form of the near optimal strategy based on counting the number of candidates after moment $n^{\ast}$. An upper bound is derived by giving DM2 additional 
information whenever he rejects a candidate after moment $n^{\ast}$. Section \ref{sec6} summarizes the results and gives directions for future research.

\section{The Auxiliary Problem in the Case $N=4$} \label{sec2}

This section first recalls the classical secretary problem. Many of the techniques described here are employed to obtain the solution of the optimization problem faced by Player 2. In addition, the version of the game outlined in Section \ref{sec1} with 
$N=4$ is considered in detail. This is the simplest version of the game illustrating the factors that need to be taken into consideration when deriving a  solution to this game for general $N$.

It is assumed that $N$ objects of differing values appear in random order. Define $A_n$ to be the absolute rank of object $n$ among all the objects. The goal of both players is to select the best object, i.e. the object with absolute rank 1. Let 
$R_n$ be the rank of the $n$-th object among the objects that have already appeared. This will be referred to as the relative rank of an object. 
Thus $R_n =1$ when the $n$-th object is the best object seen so far and 
$P(R_n =i)=\frac{1}{n}$ for $i=1,2,\ldots ,n$.  

Since Player 1 always has priority, it follows that at equilibrium she uses the optimal strategy 
in the classical secretary problem (see Gilbert and Mosteller 1966). One way of deriving this solution is by dynamic programming. 
Since the objects appear in random order, the probability that the $n$-th object is the 
best seen so far (i.e. has relative rank equal to 1) is $\frac{1}{n}$. The probability of such an object being the best overall (i.e. having absolute rank equal to 1) is $\frac{n}{N}$. Let 
$u_n$ be the optimal probability of Player 1 selecting the best object overall when she has already observed and rejected $n$ 
objects, $0\leq n\leq N$.  By definition 
$u_N=0$ and $u_0$ is the value of the game to Player 1. Player 1 only needs to 
make a decision when the current object has relative rank 1 and should accept such an object 
if and only if the probability that it is the best object overall is greater than the
probability of selecting the best object overall when search continues. Hence, 
\begin{equation}
u_{n} = \frac{1}{n+1} \max\left\{ \frac{n+1}{N},u_{n+1}\right\} + \frac{nu_{n+1}}{n+1}. \label{opt}
\end{equation}
It follows directly from this optimality equation that $u_{n}\geq u_{n+1}$. Hence, the 
equilibrium strategy of Player 1 is to accept the first object of relative rank 1 to appear from 
time $n^{\ast}$ onwards, where $n^{\ast}$ is the smallest integer satisfying 
$\frac{n^{\ast}}{N}\geq u_{n^{\ast}}$. Since $\frac{N}{N}>u_{N}=0$, $n^{\ast}$ 
exists. For $n<n^{\ast}, u_{n}=u_{n^{\ast}-1}$. 

For example, in the classical secretary problem with $N=4$, $u_{4}=0, u_{3} = \frac{1}{4}$. 
\begin{eqnarray*}
u_{2} & = & \frac{1}{3} \max \left\{ \frac{3}{4}, \frac{1}{4} \right\} +\frac{2}{3}\times \frac{1}{4} = \frac{5}{12} \\
u_{1} & = &  \frac{1}{2} \max \left\{ \frac{2}{4}, \frac{5}{12} \right\} +\frac{1}{2}\times \frac{5}{12} = \frac{11}{24} \\
u_{0} & = & \max \left\{ \frac{1}{4}, \frac{11}{24} \right\} = \frac{11}{24} 
\end{eqnarray*}
Thus the equilibrium strategy of Player 1 is to accept the first object of relative rank 1 to appear after the initial object.

One may also consider the transitions between moments at which objects of relative rank 1 appear. Such objects will be referred to as candidates according to Player 1.  These are the only moments at which Player 1 must make a decision. Let 
$\tau_j$ be the moment at which the $j$-th such candidate appears. Suppose that 
$\tau_j = n$. The next candidate appears at moment $k$, where $k>n$, if and only if the objects appearing at moments $n+1, n+2, \ldots k-1$ do not have relative rank 1, but 
the object appearing at time $k$ does. Hence, for $n<k< N$
\begin{equation}
P(\tau_{j+1}=k|\tau_j =n) = \frac{n}{n+1} \frac{n+1}{n+2} \ldots \frac{k-2}{k-1} \frac{1}{k} = 
\frac{n}{(k-1)k}. \label{tm1}
\end{equation}
It should be noted that the probability that the first candidate to appear after moment $n$ 
appears at moment $k$ is independent of whether the object to appear at moment $n$ is a 
candidate or not.  

If no candidate appears after time $\tau_j$, then by definition $\tau_{j+1}=N+1$. It should 
be noted that the event $\tau_{j+1}=N+1$ is equivalent to the event that the candidate 
appearing at time $\tau_j$ has absolute rank 1. By splitting the summand in the following equation using partial fractions, it follows that
\begin{equation}
P(A_{n}=1|R_n =1) = 1-\sum_{k=n+1}^{N} \frac{n}{(k-1)k} = \frac{n}{N}. \label{tm2}
\end{equation}

The expected reward obtained by accepting the next object of relative rank one to appear after moment $n$  is given by $z_n$, where 
\begin{equation}
z_n = \sum_{k=n+1}^N  \frac{n}{(k-1)k} \times \frac{k}{N} = \frac{n}{N} \sum_{k=n+1}^N \frac{1}{k-1} . \label{tm4}
\end{equation}
Clearly, Player 1 should not stop at time $\tau_j$ given that the expected reward from stopping at time $\tau_{j+1}$ (i.e. waiting for the next candidate to appear and accepting such a candidate) is greater. Hence, we consider a one-step look ahead rule (an OSLA rule). Under such a strategy, Player 1 accepts the first object that gives an expected reward that is at least as great as the 
expected reward from accepting the next candidate to appear. Let $B_j$ be the set of 
states in which Player 1 stops at decision point $\tau_j$ when using an OSLA rule. Chow
\textit{et al.} (1971) showed that an OSLA 
rule is optimal whenever the following condition is satisfied: if the state at time $\tau_j$
is in $B_j$, then the state at time $\tau_{j+1}$ (given that $\tau_{j+1}\leq N$) belongs to $B_{j+1}$ with probability 1. Suppose that under an OSLA rule, Player 1 should stop at time $\tau_j = n$, i.e. $z_n \leq \frac{n}{N}$. It follows that 
\begin{equation}
 \frac{n}{N} \sum_{k=n+1}^N \frac{1}{k-1}\leq \frac{n}{N} \Rightarrow 
\sum_{k=n+1}^N \frac{1}{k-1}\leq 1. \label{tm5}
\end{equation}
It is clear that if this inequality is satisfied at moment $n$, then it is satisfied at any moment 
$m>n$. It thus follows that an OSLA rule is optimal. For the secretary problem, the optimal 
strategy is to accept the first candidate to appear at or after moment $n^{\ast}$, where 
$n^{\ast}$ is the smallest value of $n$ to satisfy Inequality (\ref{tm5}). The optimal reward 
from search is given by $z_{n^{\ast}-1}$.

In order to obtain asymptotic results regarding the solution of such problems when 
$N\rightarrow \infty$, we define $t=\frac{n}{N}$ to be the proportion of objects that have been observed. Here, $t$ will be referred to as the time. The density function of the time $s$ at which the first object of relative rank 1 observed after time $t$ appears is given by $g(s; t)=\frac{t}{s^2}$. The probability that an object of relative rank 1 appearing at time $t$ has absolute rank 1 is given by
\begin{equation}
 1-\int_t^1 \frac{tds}{s^2} = t. \label{tm3}
\end{equation}
The expected reward obtained by accepting the first object of relative rank 1 to appear after time $t$ is given by $z(t)$, where
\begin{equation}
z(t) = \int_t^1 \frac{t}{s^2}\times s ds = -t\ln t. \label{tm6}
\end{equation}
The earliest time at which a candidate should be accepted, $t^{\ast}$, satisfies $t^{\ast} = -t^{\ast} \ln t^{\ast}$. Hence, $t^{\ast}=e^{-1}\approx 0.367879$. This is the optimal 
probability of accepting the absolutely best object and the optimal strategy is an OSLA rule.

Now we consider the response of Player 2 to the equilibrium strategy of Player 1.
Note that Player 2 does not observe the object chosen by Player 1, but observes all the other objects in the same sequence. Clearly, Player 2 should only consider objects that are the best that he has seen so far. Such objects will be referred to as candidates according to Player 2.

Suppose Player 2 considers the $i$-th object to appear to him to be a candidate, where 
$i<n^{\ast}$. At equilibrium, Player 2 knows that Player 1 will not have yet accepted an 
object. Hence, $R_i =1$ and Player 2 knows how many objects have appeared in total. Now assume that the $j$-th object is the first to be considered as a candidate by Player 2 after the appearance of the first $n^{\ast}-1$ objects. At equilibrium, Player 2 knows that this cannot be the first object of relative rank 1 to appear at or after moment $n^{\ast}$, since Player 1 
accepts such an object. Hence, Player 2 knows that the current object is the $j+1$-th in the original sequence. It follows that at any decision point, Player 2 can infer how many objects have appeared in total. This number will be denoted by the moment $n$. In addition, only the object obtained by Player 1 can be better than a candidate according to Player 2. Hence, if
 a candidate according to Player 2 appears at moment $n$, $n>n^{\ast}$, then 
$R_n \leq 2$.

This problem can be illustrated for the particular case $N=4$ by considering each of the $4! =24$ permutations 
representing the possible sequences of the absolute ranks of 4 objects (see Table 
\ref{tab1}). The upper index 0 signifies an object of relative rank 1 that is accepted by Player 1. The upper index $m$, $m\in \{ 1,2\}$, signifies the $m$-th candidate according to Player 2 
to appear after Player 1 has  accepted an object. Note that, at equilibrium, the first object always appears to Player 2 as a candidate. The payoffs $V_2 (\pi_k)$, $k=1,2,3$, give the payoff of Player 2 under the following strategies: $\pi_1$ - accept the first object to appear, 
$\pi_2$ - always accept the first candidate to appear after moment 2, 
$\pi_3$ - only accept a candidate at moment 4. It should be noted that the optimal response of Player 2 has to come from this set of strategies. 

The expected reward of Player 2 from using strategy $\pi_k$ is the proportion of permutations for which Player 2 selects the object with absolute rank 1. It follows from Table
\ref{tab1} that
\[
E[V_2 (\pi_1)]=1/4; \hspace{.2in} E[V_2 (\pi_2)] = 5/24; \hspace{.2in} E[V_2 (\pi_3)]=1/6.
\]
Hence, the optimal response of Player 2 to the equilibrium strategy of Player 1 is to accept the first object to appear. 

In four of the eight cases where Player 2 observes a candidate at moment 3 and 
two of the four cases  where Player 2 observes a candidate at moment 4 after not observing a candidate at time 3, such a candidate has relative rank 1 (i.e. is better than the object 
selected by Player 1). Hence, the probability that the first candidate seen by 
Player 2 after moment 2 has relative rank 1 is 1/2, regardless of when such a 
candidate appears. The second candidate to appear to Player 2 after moment 2 can only appear at 
the final moment. In two of these three cases, such a candidate has relative (and also absolute) rank 1. Hence, the probability that the second candidate to appear to Player 2 after moment 
2 has relative rank 1 is equal to 2/3. 

Suppose that Player 1 accepts an object at moment $\mu_0$. For convenience, when Player 1 does not accept an object, let $\mu_0 =N+1$. Define the moment at which the $m$-th candidate appearing to Player 2 after time $n^{\ast}$ is observed to be $\mu_m$.
If no such candidate is observed by Player 2, then let $\mu_m = N+1$. Note that 
when $\mu_1 \leq N$, then for all $\mu_0<n< \mu_1$, $R_{n}>2$ and $R_{\mu_1}\leq 2$.
 Since the objects appear in a random sequence, it follows that
$R_{\mu_1} = 1$ with probability 0.5 independently of the moment it appears. On the other hand, 
for $m\geq 2$, it will be shown in Section \ref{sec3} that the 
probability of
$R_{\mu_m}$ being equal to 1 depends on the history of the process (to be specific, on the moments at which each candidate appeared). This probability is bounded below by $\frac{m}{m+1}$ and 
always increases when a new candidate appears. 

\begin{table}[h]
\caption{Possible Sequences of Absolute Ranks with $N=4$} \label{tab1}%
\begin{tabular}{@{}lllllll@{}}
\toprule
$n=1$ & $n=2$  & $n=3$ & $n=4$ & $V_2 (\pi_1)$ & $V_2 (\pi_2)$ & $V_2 (\pi_3)$  \\
\midrule
1 & 2 & 3 & 4 & 1 & 0 & 0   \\
1 & 2 & 4 & 3 & 1 & 0 & 0 \\
1 & 3 & 2 & 4 & 1 & 0 & 0 \\
1 & 3 & 4 & 2 & 1 & 0 & 0 \\
1 & 4 & 2 & 3 & 1 & 0 & 0 \\
1 & 4 & 3 & 2 & 1 & 0 & 0 \\
2 & $1^0$ & 3 & 4 & 0 & 0 & 0 \\
2 & $1^0$ & 4 & 3 & 0 & 0 & 0 \\
2 & 3 & $1^0$ & 4 & 0 & 0 & 0 \\
2 & 3 & 4 & $1^0$ & 0 & 0 & 0 \\
2 & 4 & $1^0$ & 3 & 0 & 0 & 0 \\
2 & 4 & 3 & $1^0$ & 0 & 0 & 0 \\
3 & $1^0$ & $2^1$ & 4 & 0 & 0 & 0 \\
3 & $1^0$ & 4 & $2^1$ & 0 & 0 & 0 \\
3 & $2^0$ & $1^1$ & 4 & 0 & 1 & 0 \\
3 & $2^0$ & 4 & $1^1$ & 0 & 1 & 1 \\
3 & 4 & $1^0$ & $2^1$ & 0 & 0 & 0 \\
3 & 4 & $2^0$ & $1^1$ & 0 & 1 & 1 \\
4 & $1^0$ & $2^1$ & 3 & 0 & 0 & 0 \\
4 & $1^0$ & $3^1$ & $2^2$ & 0 & 0 & 0 \\
4 & $2^0$ & $1^1$ & 3 & 0 & 1 & 0 \\
4 & $2^0$ & $3^1$ & $1^2$ & 0 & 0 & 1 \\
4 & $3^0$ & $1^1$ & $2$ & 0 & 1 & 0  \\
4 & $3^0$ & $2^1$ & $1^2$ & 0 & 0 & 1\\  
\botrule
\end{tabular}
\footnotetext{Upper indices: 0 signifies the object selected by Player 1, $m=1,2$ signifies the numbers of candidates, as observed by Player 2, after Player 1 selects an object. $V_2$ denotes the reward of Player 2, $\pi_1$, $\pi_2$ and $\pi_3$ denote 
the strategies "accept the first object", "accept any candidate from moment 3 onwards", 
"accept a candidate only at the final moment".}
\end{table}

In general, Player 2 is in a weaker 
position than Player 1. We thus expect that Player 2 should 

\begin{description}
\item[1)] be prepared to accept a candidate at some moment $n_0$, where $n_0 <n^{\ast}$. Suppose a candidate appears to Player 2 at moment $n$, where $n\in \{ n_0, n_0+1,n_0+2,\ldots ,n^{\ast}-1\}$. Such a candidate will have relative rank 1, independently of the previous history of the game (the relative ranks of the previous objects to appear).  
\item[2)] when deciding whether to accept a candidate at moment $n$, where 
$n\geq n^{\ast}$, take into account both the number of objects seen and the probability that
the current candidate has relative rank 1. Since this probability falls from 1 to 0.5 when 
$n$ first exceeds $n^{\ast}$, it may be optimal for Player 2 to reject the first candidate to 
appear to him after moment $n^{\ast}$ (as long as it appears early enough).
\end{description}

Suppose that Player 2 observes a candidate at moment $n$, where $n>n^{\ast}$.  It follows that $R_{n} \in \{ 1,2\}$. In order to derive the distribution of $\mu_{m+1}$ given $\mu_m = n$, we need to first consider the distribution of $\mu_{m+1}$ given that the relative rank $R_n$ of the current candidate is equal to $r$, $r=1,2$. When 
$r=1$, for $\mu_m < n<\mu_{m+1}$, $R_n >1$ and $R_{\mu_{m+1}}=1$. Hence, the distribution of the moment of the appearance of the next candidate given $R_{\mu_m} =1$ is given by Equation (\ref{tm1}). When $r=2$, for $\mu_m < n<\mu_{m+1}$, $R_n >2$ and $R_{\mu_{m+1}}\leq 2$.  Arguing as in the derivation of Equation (\ref{tm1}), we obtain
\begin{equation}
P(\mu_{m+1}\! =\! k|\mu_m \! =\! n, R_n\! =\! 2) \! = \! \frac{n\! -\! 1}{n\! +\! 1} \frac{n}{n\! +\! 2} \frac{n\! +\! 1}{n\! +\! 3}\ldots \frac{k\! -\! 4}{k\! -\! 2} \frac{k\! -\! 3}{k\! -\! 1} \frac{2}{k} \! = \! \frac{2n(n - 1)}{k(k\! -\! 1)(k\! -\! 2)}. \label{tm7}
\end{equation}
Again, the distribution of the time that passes until the next object of relative rank $\leq 2$ appears is independent of the relative rank of the current object. In particular, the distribution of the moment at which the first candidate appears to Player 2 after the moment Player 1 accepts an object is also given by Equation (\ref{tm7}).

Given that the current candidate has relative rank 2, the next candidate has relative rank 1 with probability 0.5. Let $y_n$ be Player 2's expected reward from accepting the next candidate to appear after one of relative rank 2 appears at moment $n$. It follows that
\begin{equation}
y_n = \sum_{k=n+1}^N \frac{2n(n-1)}{k(k-1)(k-2)} \times \frac{k}{2N} = \frac{n(n-1)}{N} \sum_{k=n+1}^N \frac{1}{(k-1)(k-2)}. \label{yn}
\end{equation} 

Splitting the expression in Equation (\ref{tm7}) using partial fractions, it follows that
\begin{equation}
P(\mu_{m+1}\! =\! N+1|\mu_m \! =\! n, R_n\! =\! 2) \! =1- \sum_{k=n+1}^{N} \frac{2n(n - 1)}{k(k\! -\! 1)(k\! -\! 2)} = \frac{n(n-1)}{N(N-1)}. \label{tm9}
\end{equation}
This is the probability that the object observed at moment $\mu_m$ is the last candidate 
to be observed given that $R_{\mu_m}=2$.

Now we consider the asymptotic distribution of the time
 $s$ at which the first object of relative rank $\leq 2$ observed after time $t$ appears. The density function of this time is given by $h(s; t)=\frac{2t^2}{s^3}$. 
Let $y(t)$ be Player 2's expected reward from accepting the next candidate to appear after one of relative rank 2 appears at time $t$. It follows that
\begin{equation}
y(t) = \int_t^1 \frac{2t^2}{s^3} \times \frac{sds}{2} = t(1-t). \label{yt}
\end{equation}

The probability that no object appearing after time $t$ has absolute rank $\leq 2$ is given by
\begin{equation}
 1-\int_t^1 \frac{2t^2 ds}{s^3} = t^2. \label{tm8}
\end{equation}

The following section derives the form of the optimal policy for finite $N$.

\section{The Form of the Equilibrium} \label{sec3}

As stated in Section \ref{sec2}, the equilibrium strategy of Player 1 is identical to the optimal strategy of a decision maker in the classical secretary problem. Hence, solving this game theoretical model reduces to finding the optimal response of Player 2. Let $H_n$ be the history of the process according to Player 2 after observing $n$ objects when $n$ is a decision point (i.e. when the $n$-th object seen by Player 2 is the best he has seen so far). For $n<n^{\ast}$, this history can be described by the realisations of the relative ranks 
 $R_1 ,R_2,\ldots ,R_n$, since Player 2 sees all the objects up to moment $n$.
For $n\geq n^{\ast}$, as argued above, at equilibrium Player 2 can infer at any decision point that Player 1 has already obtained an object. In this case, the $n$-th object observed by Player 2 will be the $n+1$-th object in the sequence. In this case, the history $H_n$ can be described by the realisations of the relative ranks 
$S_1 ,S_2,\ldots ,S_{n}$, where $S_i$ is the rank of the $i$-th object observed by Player 2 
with respect to the objects previously observed by him. Suppose Player 1 accepts a candidate at moment $\mu_0$, where $\mu_0$ is the first moment $n\geq n^{\ast}$ such that 
$R_{\mu_0}=1$. By definition, for $n<\mu_0$, 
$S_n =R_n$. For $n>\mu_0$, since Player 2 sees all of the objects except for the one selected by Player 1, $S_{n}$ is either equal to $R_{n+1}$ (when the object observed by 
Player 2 is better than the object selected by Player 1) or $R_{n+1}+1$ (when the object observed by Player 2 is worse than the object selected by Player 1).
To find the optimal response of Player 2, we should derive  $P(R_{M(n)} = 1|S_{n}=1,H_{n})$, where $M(n)$ is the moment at which Player 2 observes the $n$-th object seen by him. As argued above, for $n\leq n^{\ast}-1$, at equilibrium $M(n)=n$ with probability 1 and 
  $P(R_n = 1|S_n=1,H_{n})=1$. Also, when any candidate appears to Player 2 after moment $n^{\ast}-1$, then $M(n)=n+1$ with probability 1. Suppose that the first such candidate is the $n$-th object to be seen by Player 2. It was argued in Section \ref{sec2} that 
\[
P(R_{n+1}=1|S_n =1, H_{n})=P(R_{\mu_1}=1|H_{\mu_1}) = 1/2.
\]
It should be stressed that $\mu_m$ is the moment at which the $m$-th such candidate appears. Hence, in the equation above $\mu_1 = n+1$.
Note that the moment at which the $m$-th such candidate appears satisfies 
$\mu_m \geq n^{\ast}+m$. For $n^{\ast}+m\leq \mu_m \leq N$, define $P(R_{\mu_m}=1|H_{\mu_m})=p_{m,\mu_m,H}$. It follows that $P(R_{\mu_m}=2|H_{\mu_m})=1-p_{m,\mu_m,H}$. Although the value of $m$ can be inferred from the history of the process, $m$ is used as a separate argument in order to more easily compare the optimal strategy with the near-optimal strategy derived in Section \ref{sec4}.

Given that $\mu_m =n$, the distribution of the time of the appearance of the next candidate 
to Player 2 can be derived by applying the law of total probability to Equations (\ref{tm1}) and (\ref{tm7}), i.e.
\begin{equation}
P(\mu_{m+1}=k|\mu_m=n,\! H_{n-1})  =  \frac{np_{m,n,H}}{k(k-1)} + \frac{2n(n-1)(1-p_{m,n,H})}{k(k-1)(k-2)}. \label{totprob} 
\end{equation}
By definition,
\begin{equation}
P(R_{\mu_{m+1}}\! \! =\! 1|\mu_{m+1}\! =\! k, \mu_{m}\! =\! n, H_{n-1}) \! = \!  \frac{P(R_{\mu_{m+1}}\! =\! 1,\mu_{m+1}\! =\! k, \mu_{m}\! =\! n, H_{n-1})}{P(\mu_{m+1}=k, \mu_{m}=n, H_{n-1})}.
\label{cond1}
\end{equation}
In addition,
\begin{eqnarray}
P(\mu_{m+1}\! \! =\! \! k, \mu_{m}\! \! =\! \! n, H_{n-1}) & \! \! = \! \! & P(\mu_{m+1}=k|\mu_{m}=n, H_{n-1})P(\mu_m =n,H_{n-1}) \nonumber \\
& \! \! = \! \! & \left[ \frac{np_{m,n,H}}{k(k-1)} \! + \! \frac{2n(n\! -\! 1)(1\! -\! p_{m,n,H})}{k(k-1)(k-2)} \right] \! \! 
P(\mu_m\! =\! n, \! H_{n-1}). \label{joint1}
\end{eqnarray} 
Since $R_{\mu_m}$ can only take the values 1 and 2, we obtain
\begin{eqnarray}
P(R_{\mu_{m+1}}\! \! =\! \! 1,\mu_{m+1} & \! \!  = \! \!  & k, \mu_{m}\! =\! n, \! H_{n-1}) \! \! = \!  \! \sum_{i=1}^2 \! \! P(R_{\mu_{m+1}}\! \! =\! \! 1,\mu_{m+1}\! \! =\! \! k, R_{\mu_m}\! \! =\! \! i, \mu_{m}\! \! =\! \! n, H_{n-1}) \nonumber \\
 &  \! \!  =   \! \!  &  \sum_{i=1}^2 \left[ P (R_{\mu_{m+1}} = 1,\mu_{m+1} = k | R_{\mu_m} = i, \mu_{m} = n, H_{n-1})\times  \ldots \right. \nonumber \\
  & \! \! \! \!   & \left. \times P(R_{\mu_m}=i|\mu_m =n, H_{n-1}) P(\mu_m =n,H_{n-1}) \right] 
\nonumber \\
 &   \! \!  =   \! \! &   \left[ \frac{np_{m,n,H}}{k(k-1)}  + \frac{n(n - 1)(1 - p_{m,n,H})}{k(k-1)(k-2)} \right] \! P(\mu_m \! = \! n,H_{n-1}).  \! \! \! \!  \! \! \! \! 
 \label{joint2}
\end{eqnarray} 
From Equations (\ref{joint1}) and (\ref{joint2}), it follows that
\begin{eqnarray} 
 p_{m+1,k,H} & = & P (R_{\mu_{m+1}} = 1|\mu_{m+1} = k,\mu_m  = n, H_{n-1}) \nonumber \\
& = & \frac{(k - 2)p_{m,n,H} + (n - 1)(1 - p_{m,n,H})}{(k - 2)p_{m,n,H} + 2(n - 1)(1 - p_{m,n,H})} . \nonumber \\
& = & 1 - \frac{(n-1)(1-p_{m,n,H})}{(k-2)p_{m,n,H}+2(n-1)(1-p_{m,n,H})}. \label{induc1}
\end{eqnarray}
Since $p_{1,n,H}=1/2$ for all $n^{\ast}<n<N$, the value of $p_{m,k,H}$ for any given 
realization of the process can be calculated inductively. It can be seen that a list of the 
decision points of Player 2 $(\mu_1 ,\mu_2, \ldots ,\mu_{m-1}$ where $\mu_1 >n^{\ast})$ is a rich enough history to calculate $p_{m,n,H}$. For mathematical convenience, we set 
$p_{0,n,H}=0$ for all $n$. This conforms to the inductive definition given above.

For integers $n^{\ast}\leq n<k\leq N$ and $p\in [0,1]$ Define 
\begin{equation}
f(n,k,p) = 1 - \frac{(n-1)(1-p)}{(k-2)p+2(n-1)(1-p)}.
\end{equation}

\begin{proposition} \label{prop1}
\begin{description}
Given $\mu_m =n$ and $\mu_{m+1}=k>n$:
\item[a)] $p_{m+1,k,H}>p_{m,n,H}$.
\item[b)]  For fixed $n$, $p_{m+1,k,H}$ is 
strictly increasing in $k$.
\item[c)] For fixed $k$ and $n$, $f(n,k,p)$ is strictly increasing in $p$.
\item[d)] The minimum value of $p_{m,n,H}$ is equal to $\frac{m}{m+1}$ and is only attained when  the first $m$ decision points of Player 2 after moment $n^{\ast}$ occur at successive moments.
\end{description}
\end{proposition}

\begin{proof}
To prove part a), note that
\begin{eqnarray*}
p_{m+1,k,H} - p_{m,n,H} & = & \frac{(k - 2)p_{m,n,H} + (n - 1)(1 - p_{m,n,H})}{(k - 2)p_{m,n,H} + 2(n - 1)(1 - p_{m,n,H})} - p_{m,n,H} \\
& = & \frac{(k-n-1)(p_{m,n,H}-p_{m,n,H}^2)+(n-1)(1-p_{m,n,H})^2}{(k - 2)p_{m,n,H} + 2(n - 1)(1 - p_{m,n,H})}. 
\end{eqnarray*}
This expression is positive, since by definition $k\geq n+1$.

Part b) follows directly from the fact that the denominator of the fraction in Equation 
(\ref{induc1}) is clearly increasing in $k$, while the numerator is fixed.

Part c) follows from the fact that
\[
\frac{\partial f}{\partial p} = \frac{(n-1)(k-2)}{[(k-2)p+2(n-1)(1-p)]^2}>0.
\]

To prove part d), assume that the theorem holds for $m=j$. 
From  Equation (\ref{induc1}), together with parts b) and c) of the propostion, it follows that when $p_{j,n,H}=\frac{j}{j+1}$, then the minimum value of $p_{j+1,k,H}$ is attained when $k=n+1$ and 
\[ 
p_{j+1,n+1,H} = 1 - \frac{(n-1)/(j+1)}{(n-1)j/(j+1)+2(n-1)/(j+1)} = \frac{j+1}{j+2}.
\]
Hence, the theorem follows for $m=j+1$. The theorem then follows by induction from the fact that $p_{1,n,H}=1/2$ . 
\end{proof}

Let the state of Player 2 on observing the $m$-th candidate seen after moment $n^{\ast}$ 
at moment $n$ to be $(n,p_{m,n,H})$. Define $v(n,p_{m,n,H})$ to be the optimal expected 
payoff of Player 2 after rejecting a candidate in this state. Given that the current candidate 
has relative rank 1, it is the best of all the objects with probability $n/N$. Otherwise, it cannot be the best of all the objects. It follows 
that Player 2 should accept such a candidate if and only if 
\[
v(n,p_{m,n,H}) \leq \frac{n p_{m,n,H}}{N}.
\]
This state space is rather complex, especially when $N$ is large. However, $p_{m,n,H}$ can only take certain values 
in the interval $[0.5, 1]$ [for example $p_{1,n,H}=1/2, p_{2,n,H}\geq 2/3$]. For convenience, for $n\geq n^{\ast}$ define $v(n,0)$ to be the optimal future expected reward of Player 2  after $n$ objects have appeared, Player 1 has already accepted an offer, but Player 2 has not yet observed a candidate after that moment. Analogously, for $1\leq n\leq N$ we define 
$v(n,-1)$ to be the optimal future expected reward of Player 2 when neither player has accepted any of the first $n$ objects. For $n\geq n^{\ast}$, Player 2 cannot distinguish between states $(n,-1)$ and $(n,0)$. However, Player 2 can infer the state at any of his decision points. 

From the above definitions, it follows that $v(N,p)=v(N-1,-1)=0$ for any $p\in \{ -1\} \cup [0,1]$.  The value of the game to Player 2 is given by $v(0,-1)$. For $n> n^{\ast}$ and 
$p_{m,n,H}> 0$, the optimality equation is given by 
\begin{equation}
v(n,p_{m,n,H}) = E\left[ \mathcal{I}_{k\leq N} \max \left\{ \frac{kp_{m+1,k,H}}{N}, v(k,p_{m+1,k,H})\right\} \right], \label{opt1p2}
\end{equation}
where $k$ is the moment at which the next candidate appears to Player 2 (or $N+1$ when no such candidate appears) and $\mathcal{I}$ is the indicator function. Since Player 1 accepts the first object of relative rank 1 to appear after moment $n^{\ast}-1$, then by conditioning on the
time of the appearance of such an object, it follows that for $n\geq n^{\ast}-1$
\begin{equation}
v(n,-1) = \sum_{k=n+1}^N \frac{nv(k,0)}{k(k-1)}. \label{opt2p2}
\end{equation}
For $n< n^{\ast}-1$, conditioning on whether an object of relative rank 1 appears before moment $n^{\ast}$ or not, the optimality equation is given by
\begin{equation}
v(n,-1) = E\left[ \mathcal{I}_{k<n^{\ast}} \max \left\{ \frac{k}{N},v(k,-1) \right\} \right] 
+v(n^{\ast}-1,-1) E[\mathcal{I}_{k\geq n^{\ast}}]. \label{opt3p2}
\end{equation}
The dependence of the optimality equations on $p_{m,n,H}$ significantly complicates the 
general form of the optimal response of Player 2. For example, suppose Player 2 
observes a candidate at moment $n^{\ast}-1$. Such a candidate will definitely have relative rank 1 and is thus the best of all objects with probability 
$\frac{n^{\ast}-1}{N}$. At equilibrium, Player 1 will still be searching at this moment. Now consider the situation in which Player 2 observes a candidate at moment $n^{\ast}+1$. Such a candidate has relative rank 1 with probability 0.5 and is thus the best of all objects with probability
$\frac{n^{\ast}+1}{2N}$. Unless $N$ is very small, this will be smaller than $\frac{n^{\ast}-1}{N}$. In addition, given such a candidate appears, at equilibrium Player 2 knows that Player 1 has already accepted an object and thus Player 2 does not face any competition in choosing an object. For these reasons, when $N$ is reasonably large, intuitively Player 2 should accept a candidate at moment $n^{\ast}-1$, but reject a candidate at moment $n^{\ast}+1$.

We start by deriving the form of the optimal strategy of Player 2 for $n> n^{\ast}$. 
Note that by stopping in state $(n,p_{m,n,H})$, Player 2 obtains an expected reward of $\frac{np_{m,n,H}}{N}$. Define the expected reward from stopping at the next decision point to be $w(n,p_{m,n,H})$.  Using the law of total probability, 
\begin{eqnarray}
w(n,p_{m,n,H}) \! \! \! & = & \! \! p_{m,n,H}z_n +(1-p_{m,n,H})y_n \nonumber \\
\! \! \! & = & \! \! \frac{np_{m,n,H}}{N} \! \! \left[ \sum_{k=n+1}^N \! \! \frac{1}{k\! -\! 1} \! \right]  \! + \!
\frac{(1\! -\! p_{m,n,H})n(n\! -\! 1)}{N} \! \! \left[ \sum_{k=n+1}^N \! \! \! \frac{1}{(k\! -\! \! 1)(k\! -\!2)}\! \right] \! \! . 
\label{opt4p2}
\end{eqnarray}
Using an OSLA rule, Player 2 accepts an object when $\frac{np_{m,n,H}}{N} \geq 
w(n, p_{m,n,H})$. This condition is equivalent to
\begin{equation}
p_{m,n,H} \! \geq \! q_n \! \! = \! \frac{S_{2,n}}{\!1\! +\! S_{2,n}\! -\! S_{1,n}\!}, \mbox{where } 
S_{1,n} \! = \! \! \! \! \! \sum_{k=n+1}^N  \! \! \frac{1}{k \! - \! 1\! } \mbox{ and } S_{2,n} \! = \! \! \! \! \! \sum_{k=n+1}^N  \! \!
\frac{n-1}{(k\!  -\! 1)(k\! -\! 2)\!} . \label{cond11}
\end{equation}

\begin{proposition} \label{prop2}
The term $q_n$ is decreasing in $n$ for $n> n^{\ast}$.
\end{proposition}

\begin{proof}
We have
\[
q_n - q_{n+1} = \frac{S_{2,n}}{1+S_{2,n}-S_{1,n}} - \frac{S_{2,n+1}}{1+S_{2,n+1}-S_{1,n+1}} 
\]
By definition $S_{2,n}>0$ and for $n>n^{\ast}$, $S_{1,n}<1$. Hence, the numerators and 
the denominators in the fractions above are all positive.
It thus suffices to show that the numerator in the first fraction, $S_{2,n}$, is greater than the numerator in the second fraction, $S_{2,n+1}$, and the denominator in the first fraction
is smaller than the denominator in the second fraction.
Note that $S_{1,n}-S_{1,n+1} =1/n$. In addition, 
\begin{eqnarray*} 
S_{2,n} - S_{2,n+1} & = & \sum_{k=n+1}^N \frac{n-1}{(k-1)(k-2)} - \sum_{k=n+2}^N \frac{n-1}{(k-1)(k-2)}- \sum_{k=n+2}^N \frac{1}{(k-1)(k-2)} \\
& = & \frac{1}{n} - \sum_{k=n+2}^N \frac{1}{(k-1)(k-2)}.
\end{eqnarray*}
Thus
\[
1+S_{2,n}-S_{1,n} = 1+S_{2,n+1}-S_{1,n+1}-\sum_{k=n+2}^N \frac{1}{(k-1)(k-2)}
< 1+S_{2,n+1}-S_{1,n+1}.
\]
In addition, 
\begin{eqnarray*} 
\sum_{k=n+2}^N \frac{1}{(k-1)(k-2)} & \leq &  \sum_{k=n+2}^N \frac{1}{(k-1)n} \\
& \leq &  \frac{1}{n} \sum_{k=n+2}^N \frac{1}{k-1} <  \frac{1}{n}.
\end{eqnarray*} 
The final inequality results from the fact that for $n>n^{ \ast}$,  $\sum_{k=n+1}^N \frac{1}{(k-1)}<1$. It follows that $S_{2,n}>S_{2,n+1}$.
\end{proof}

\begin{lemma}[to Proposition \ref{prop2}] \label{lemma1}
Suppose that for $n>n^{\ast}$ under an OSLA rule Player 2 accepts 
an object in state $(n,p_{m,n,H})$ when $(n,p_{m,n,H})\in B_n$. Let the state at the next decision point be $(k,p_{m+1,k,H})$. It follows from Proposition \ref{prop2} that $(k,p_{m+1,k,H})\in B_k$, since $k>n$ 
and $p_{m+1,k,H}>p_{m,n,H}$. Hence, an OSLA rule is optimal for $n\geq n^{\ast}$.  
\end{lemma}

It follows that the optimal strategy for $n> n^{\ast}$ is defined by a set of decreasing thresholds $q_n$, such that Player 2 should accept a candidate at moment $n$ when 
$p_{m,n,H}\geq q_n$. Since for a given realisation of the sequence of offers the appropriate $p_{m,n,H}$ can be calculated iteratively, it is relatively simple to implement the optimal 
strategy based on this criterion. Calculation of the optimal expected reward is more complex.

\textbf{Note}: For $n<n^{\ast}$ it is expected that the sufficient condition for an OSLA rule to be optimal given above will not be satisfied. For example, it was argued that it is optimal for Player 2 to accept a candidate just before moment $n^{\ast}$, but reject such a candidate just after moment $n^{\ast}$. For $n<n^{\ast}$ the optimal strategy of Player 2 will be derived using the fact that his expected reward from future search 
is decreasing in $n$ (see Example 1).

For $p_{m,n,H}<q_n$, we must consider two possibilities. Firstly, if the next candidate 
appears at moment $k$, where $k<k_0$ and $k_0$ is the smallest natural number $k$ 
satisfying the inequality
\[
p_{m+1,k,H} = \frac{p_{m,n,H}(k-2)+(1-p_{m,n,H})(n-1)}{p_{m,n,H}(k-2)+2(1-p_{m,n,H})(n-1)} \geq q_k,
\]
then Player 2 has an expected reward of $v(k,p_{m+1,k,H})$. Otherwise, Player 2 will 
accept the next candidate to appear. Considering the arrival time of the next candidate, it follows that
\begin{eqnarray}
v(n,p_{m,n,H}) & = & \left\{ \sum_{k=n+1}^{k_0-1} \left[ 
\frac{np_{m,n,H}}{k(k-1)} + \frac{2n(n-1)(1-p_{m,n,H})}{k(k-1)(k-2)} \right] 
v(k,p_{m+1,k,H}) \right\} + \nonumber \\
& & + \sum_{k=k_0}^N \left[ \frac{np_{m,n,H}}{N(k-1)} + \frac{(1-p_{m,n,H})n(n-1)}
{N(k-1)(k-2)} \right] . \label{opt5p2} 
\end{eqnarray}
For $n<n^{\ast}$, $v(n,-1)$ is non-increasing in $n$, since for $n<k<n^{\ast}$, Player 2 
can ensure an expected payoff of $v(k,-1)$ by rejecting objects $n+1, n+2, \ldots ,k$ and 
then using the optimal response. The expected payoff from stopping at moment 
$n$ in this case is $n/N$, which is increasing in $n$. Hence, the optimal response for 
$n<n^{\ast}$ is of the form "accept a candidate iff $n\geq n_0$". Suppose $n_0 -1 \leq n < n^{\ast}-1$. If the next candidate appears before time $n^{\ast}$, then Player 2 accepts such a candidate. Otherwise, Player 2 must wait until Player 1 accepts the next candidate to 
appear before reaching a decision point. The probability that no such candidate appears by 
time $n^{\ast}-1$ is $\frac{n}{n^{\ast} -1}$.
Hence, considering the time at which the next candidate appears
\begin{equation}
v(n,-1) = \frac{n}{N} \left[ \sum_{k=n+1}^{n^{\ast}-1} \frac{1}{k-1} \right] 
+ \frac{n}{n^{\ast}-1} v(n^{\ast} -1,-1). \label{opt6p2} 
\end{equation}
The optimality condition states that such a candidate should be accepted if and only if 
$n/N \geq v(n, -1)$. Hence, $n_0$ is the smallest value of $n$ that satisfies the 
inequality
\begin{equation}
1 \geq \frac{Nv(n^{\ast}-1,-1)}{n^{\ast}-1} + \sum_{k=n+1}^{n^{\ast}-1} \frac{1}{k-1}.
\label{opt7p2} 
\end{equation}
For $0\leq n<n_0$, $v(n,-1) = v(n_0-1,-1)$. 

\begin{example} \label{ex1}
This example shows how these equations are used in the case $N=50$. The thresholds $q_n$ are interpreted in such a way as to present the optimal policy in a form based simply 
on the arrival times of the candidates. First we consider the equilibrium strategy of Player 1. 
We have
\[
\sum_{k=19}^{50} \frac{1}{k-1} >1 \mbox{ and } \sum_{k=20}^{50} \frac{1}{k-1} <1.
\]
Hence, $n^{\ast}=19$ and the equilibrium strategy of Player 1 is to reject the first 18 
objects and accept the first candidate to appear afterwards. The equilibrium payoff of 
Player 1 is given by 
\[
z_{18} = \frac{18}{50} \sum_{k=19}^{50} \frac{1}{k-1} \approx 0.374275.
\]

Now we derive the optimal response of Player 2 for $n> n^{\ast}$.
The optimality criterion at moment $n$ is to accept a candidate in state $(n,p_{m,n,H})$ if and only if 
$p_{m,n,H}\geq q_n$, where $q_n$ is given by Equation (\ref{cond11}). 
Note that 
$p_{1,n,H}=1/2$ for any $n\geq n^{\ast}+1$, $p_{m,n,H}\geq \frac{m}{m+1}$ for 
$n\geq n^{\ast}+m$ and $q_n$ is decreasing in $n$. Hence, if $n_1$ is the smallest 
integer value of $n$ that satisfies the inequality $q_n \leq 0.5$, then any
candidate appearing to Player 2 at moment $n_1$ or later should be accepted. 
The values of $q_n$ relevant to the derivation of the optimal strategy are given in Table
\ref{q} to four decimal places.

\begin{table}[h]
\caption{Threshold values, $q_n$, defining the optimal response of Player 2 for $N=50$.}\label{q}%
\begin{tabular}{@{}lllllllllll@{}}
\toprule
$n$ & 20 & 21 & 22 & 23 & 24 & 25 & 26 & 27 & 28 & 29 \\
\midrule
$q_n$ & 0.8993 & 0.8331 & 0.7747 & 0.7225 & 0.6753 & 0.6323 & 0.5926 & 0.5558 & 
0.5213 & 0.4889 \\
\botrule
\end{tabular}
\footnotetext{Source: Author's calculations}
\end{table}

It follows that any candidate appearing to Player 2 at or after moment 29 should be accepted. 
Since $p_{2,n,H}\geq 2/3$, it follows that the second candidate to appear to Player 2 after 
moment $n^{\ast}=19$ should always be accepted at or after moment 25. It remains to state whether such a candidate should be accepted at an earlier moment. The 
probability that such a candidate has relative rank equal to one depends on the moments at which the first two candidates to appear to Player 2 after moment $n^{\ast}=19$ are observed. The relevant
probabilities are given in Table \ref{p} to four decimal places. Comparing these values with the relevant threshold for $\mu_2$, it can be seen that if the second candidate to be seen by 
Player 2 after moment 19 appears at moment 24, then it should only be accepted when the 
first candidate appeared at moment 20 or 21.

\begin{table}[h]
\caption{Probabilities that the candidate appearing to Player 2 at moment 
$\mu_2$ has relative rank 1 conditional on $\mu_1$ and $\mu_2$ $(N=50)$.}\label{p}%
\begin{tabular}{@{}lllll@{}}
\toprule
& $\mu_2 =21$ & $\mu_2=22$ & $\mu_2=23$ & $\mu_2=24$ \\
\midrule
$\mu_1 =20$ & 0.6667 & 0.6724 & 0.6780 & 0.6833 \\
$\mu_1 =21$ & - & 0.6667 & 0.6721 & 0.6774 \\
$\mu_1 =22$ & - & - & 0.6667 & 0.6719 \\
$\mu_1 =23$ & - & - & - & 0.6667 \\
\botrule
\end{tabular}
\footnotetext{Source: Author's calculations}
\end{table}

Since $p_{3,n,H}\geq 3/4$, it follows that the third candidate to appear to Player 2 after 
moment $n^{\ast}=19$ should always be accepted at or after moment 23. Note that such a 
candidate cannot appear earlier than at moment 22 and in this case $p_{3,22,H}=3/4$. Hence, comparison with $q_{22}$ indicates that such a candidate should not be accepted at moment 22. Finally, since $p_{4,n,H}\geq 4/5$, it follows that the fourth candidate to appear to Player 2 after 
moment $n^{\ast}=19$ should always be accepted, since it cannot appear before moment 23. 

To summarise, when $n>n^{\ast}=19$ Player 2 should:
\begin{description}
\item[i)] accept the first candidate seen after moment 19, if and only if $n\geq 29$,
\item[ii)] accept the second candidate seen after moment 19, if $n\geq 25$ or 
$n=24$ and the first candidate seen after moment 19 was seen before moment 22,
\item[iii)] accept the third candidate seen after moment 19, if and only if $n\geq 23$,
\item[iv)] always accept the fourth candidate seen after moment 19.
\end{description}

We now derive the expected future reward starting in state 
$(n,p_{m,n,H})$, where $n^{\ast}< n\leq 50$
and $p_{m,n,H}\in [0,1]$. From the nature of the problem, $p_{m,n,H}$ can only take 
a finite set of values in the interval $[0,1]$ while an optimally reacting Player 2 is still searching. Based on the optimal strategy and the possible appearance times of candidates 
$p_{m,n,H} \in \{ 0, 1/2, 2/3, p_{A}\approx 0.6719, p_{B}\approx 0.6721, p_{C} \approx 0.6724, p_{D}\approx 0.6780, 3/4\}$. Note that $p_{m,n,H}$ can only attain a value of $\geq p_A$ while an optimally acting Player 2 is still searching when $n\geq 22$. Starting from the state 
$(n,p_{m,n,H})$, where $p_{m,n,H}\geq p_A$ and $n\geq 22$, it is optimal for Player 2 to accept the next candidate to appear. From Equation (\ref{opt4p2}), it thus follows that for 
$p\geq p_A$ and $n\geq 22$
\begin{equation}
v(n,p) = \frac{np}{50} \left[ \sum_{k=n+1}^{50}  \frac{1}{k - 1} 
\right]  + \frac{(1 - p)n(n - 1)}{50} \left[ \sum_{k=n+1}^{50} \frac{1}{(k - 1)(k - 2)} 
 \right] . \label{exam1}
\end{equation}
Hence, for $p_{m,n,H} \in \{ p_A, p_B, p_C, p_D, 3/4\}$, the relevant values of the future expected reward can be calculated directly.

Note that when $p_{m,n,H}=2/3$ and $n\geq 22$, Player 2 should accept the next candidate to appear. Hence, the appropriate values of the future expected reward can be 
calculated directly by substituting $p=2/3$ into Equation (\ref{exam1}). Since $p$ can only be equal to 2/3 for $n\geq 21$, it remains to calculate $v(21,2/3)$. Starting from the state $(21, 2/3)$, the next candidate to appear should be accepted from moment 23 onwards. Given that a candidate appears at moment 22, then the probability that it has relative rank 1 is equal to 3/4. Hence, from Equation (\ref{opt5p2})
\begin{equation}
v(21, 2/3) = \frac{2v(22,3/4)}{33} + \frac{21}{75} \left[ \sum_{k=23}^{50} \frac{1}{k-1} \right] +\frac{7}{5} 
\left[ \sum_{k=23}^{50} \frac{1}{(k-1)(k-2)} \right] . 
\label{exam2}
\end{equation}
Since $v(22,3/4)$ has already been derived, it follows that the relevant values of the future expected reward function can be calculated when $p_{m,n,H}=2/3$.

Similarly, when $p_{m,n,H}=1/2$ and $n\geq 24$, Player 2 should accept the next candidate to appear. Hence, the appropriate values of the future expected reward can be 
calculated directly by substituting $p=1/2$ into Equation (\ref{exam1}). 

For $n\in \{ 20,21\}$, starting from the state $(n, 1/2)$, it follows that the next candidate to appear to Player 2 should be accepted if and only if it appears at moment 24 or later. Hence, 
 from Equation (\ref{opt5p2})
 \begin{eqnarray}
v(n,1/2) & = & \frac{n}{2}\left[ \sum_{k=n+1}^{23} \frac{v(k,p_{2,k,H})}{k(k-1)} \right] 
+ n(n-1)\left[   \sum_{k=n+1}^{23} \frac{v(k,p_{2,k,H})}{k(k-1)(k-2)} \right] +\nonumber \\
& & + \frac{n}{100}\left[ \sum_{k=24}^{50} \frac{1}{k-1} \right] + 
\frac{n(n-1)}{100} \left[ \sum_{k=24}^{50} \frac{1}{(k-1)(k-2)} \right] \label{exam3} . 
\end{eqnarray}
It should be noted that since $p_{2,k,H}\in \{ 2/3, p_B\approx 0.6721, p_C\approx 0.6724,p_D \approx 0.6780\}$, 
the values $v(k,p_{2,k,H})$ for $k\in \{21,22,23 \}$ have already been derived. 

For $n\in \{ 22,23\}$, starting from the state $(n, 1/2)$, it follows that the next candidate to appear to Player 2 should be accepted if and only if it appears at moment 25 or later. Hence, 
 from Equation (\ref{opt5p2})
 \begin{eqnarray}
v(n,1/2) & = & \frac{n}{2}\left[ \sum_{k=n+1}^{24} \frac{v(k,p_{2,k,H})}{k(k-1)} \right] 
+ n(n-1)\left[   \sum_{k=n+1}^{24} \frac{v(k,p_{2,k,H})}{k(k-1)(k-2)} \right] +\nonumber \\
& & + \frac{n}{100}\left[ \sum_{k=25}^{50} \frac{1}{k-1} \right] + 
\frac{n(n-1)}{100} \left[ \sum_{k=25}^{50} \frac{1}{(k-1)(k-2)} \right] \label{exam4} . 
\end{eqnarray}
 Note that since $p_{2,k,H}\in \{ 2/3, p_A\approx 0.6719 \}$, 
the values $v(k,p_{2,k,H})$ for $k\in \{23,24 \}$ have already been derived. 
It follows that all the relevant values $v(n,1/2)$ can be derived.

Starting from state $(n,0)$ for $n\geq 28$, Player 2 should accept the next candidate to appear. Hence, the appropriate values of the future expected reward can be 
calculated directly by substituting $p=0$ into Equation (\ref{exam1}). 

For $19\leq n\leq 27$, by considering the moment at which the first candidate appears to Player 2 after moment $n^{\ast}=19$, from Equation (\ref{opt5p2})  
\begin{equation}
v(n,0) = 2n(n-1)\left[ \sum_{k=n+1}^{28} \frac{v(k,1/2)}{k(k-1)(k-2)}\right] + \frac{n(n-1)}{50} \left[ \sum_{k=29}^{50} \frac{1}{(k - 1)(k-2)} 
\right] .\label{exam5}
\end{equation}
Since the relevant values of $v(k,1/2)$ have already been derived, it follows that all the 
values $v(n,0)$ can be derived for $n\geq 19$.

Considering the distribution of the time at which Player 1 accepts a candidate when still searching at moment $n$, where $n\geq 18$, from 
Equation (\ref{opt2p2})
\begin{equation}
v(n,-1) = \sum_{k=n+1}^{50} \frac{nv(k,0)}{k(k-1)} \label{exam6} .
\end{equation}
Since the relevant values of $v(k,0)$ have already been derived, the values of the function $v(n,-1)$ for $n\geq 18$ can thus be calculated. It can be shown 
by direct calculation that $v(18,-1)\approx 0.145870$.   

Now we consider the period when Player 2 should preempt Player 1 by accepting a candidate, i.e. $n_0 -1 \leq n\leq n^{\ast}-1=18$. From Equation (\ref{opt6p2}),
\begin{equation}
v(n,-1) = \frac{n}{N} \left[ \sum_{k=n+1}^{18} \frac{1}{k-1} \right] 
+ \frac{n}{18} v(18,-1). \label{exam7}
\end{equation} 
 From Condition (\ref{opt7p2}) Player 2 should accept a candidate appearing at moment $n$, $n\leq 18$, if and only if 
\[
\frac{25v(18,-1)}{9} + \sum_{k=n+1}^{18} \frac{1}{k-1} \leq 1.
\]
It follows that $n_0 =11$. Hence, from Equation (\ref{opt6p2}) the value of the game to Player 2 is given by 
\[
v(0,-1) = v(10,-1) = \frac{5v(18,-1)}{9} + \frac{1}{5} \sum_{k=11}^{18} \frac{1}{k-1} 
\approx 0.203157.
\]
To summarise, the optimal response of Player 2 is of the following form:
\begin{description}
\item[i)] accept a candidate if $11\leq n\leq 18$, 
\item[ii)] accept the first candidate seen after moment 19, if and only if $n\geq 29$,
\item[iii)] accept the second candidate seen after moment 19, if $n\geq 25$ or 
$n=24$ and the first candidate seen after moment 19 was seen before moment 22,
\item[iv)] accept the third candidate seen after moment 19, if and only if $n\geq 23$,
\item[v)] always accept the fourth candidate seen after moment 19.
\end{description}
\end{example} 

Such a definition of the optimal strategy will be described as \textit{a priori}, since it gives precise instructions to a DM on how to behave before any objects in the sequence have been observed. Intuitively, as $N$ increases, it becomes more difficult to make such a
description. However, it is always possible to implement the optimal strategy by 
updating the probability that a candidate according to Player 2 has relative rank 1 every time that such a candidate appears, i.e. using an online procedure.

Additionally, the solution to this game was found for $3\leq N\leq 10$ and $N=20$. 
Player 1 should accept a candidate if and only if $n\geq n^{\ast}$. The value of the game to 
Player $i$ is given by $u_i$.
For $N=3,4$, Player 2 should always accept the first object. 
For $N\geq 5$, the best response of Player 2 can be described by two thresholds, $n_0$ and $n_1$. Player 2 should accept a candidate when $n_0 \leq n<n^{\ast}$ and should accept the first candidate seen after moment $n^{\ast}$ if and only if $n\geq n_1$. The second 
candidate to appear after moment $n^{\ast}$ should always be accepted. These equilibrium 
strategies are described in Table (\ref{tab2c}).

\begin{table}[h]
\caption{Equilibrium strategies and values}\label{tab2c}%
\begin{tabular}{@{}llllll@{}}
\toprule
$n$ & $n^{\ast}$ & $u_1$ & $n_0$ & $n_1$ & $u_2$ \\
\midrule
3 & 2 & 0.5000 & 1 & - & 0.3333 \\
4 & 2 & 0.4583 & 1 & - & 0.2500 \\
5 & 3 & 0.4333 & 2 & 4 & 0.2667 \\
6 & 3 & 0.4278 & 2 & 4 & 0.2472 \\ 
7 & 3 & 0.4143 & 2 & 5 & 0.2337 \\
8 & 4 & 0.4098 & 2 & 5 & 0.2348 \\
9 & 4 & 0.4060 & 2 & 6 & 0.2199 \\
10 & 4 & 0.3987 & 3 & 6 & 0.2153 \\
20 & 8 & 0.3842 & 5 & 12 & 0.2095 \\
\botrule
\end{tabular}
\footnotetext{Source: Author's calculations}
\end{table}

It is interesting to note that the value of the game to Player 2 is not decreasing in $N$.
Increases in this value tend to coincide with increases in $n^{\ast}$, e.g. when 
$N=5$ and $N=8$. This is due to the fact that it becomes more likely that a Player 2 
accepts a candidate just before moment $n^{\ast}$. This is an advantageous situation for 
Player 2.

\section{A Near-Optimal Solution} \label{sec4}

Even when $N$ is large, Player 2 is very likely to only see a small number of candidates 
after moment $n^{\ast}$ before choosing an object. The probability of the first candidate to appear to Player 2 after 
moment $n^{\ast}$ having relative rank one is always 1/2. If Player 2 observes several such candidates without accepting one, these candidates will appear in close succession. Hence, 
$p_{m,T_m ,H}\approx \frac{m}{m+1}$ for $m=2,3,\ldots$. Hence, a good approximation to the optimal rule can be obtained by setting $p_{m,T_m ,H}=\frac{m}{m+1}$. 
It follows that the state space at decision points where $n\geq n^{\ast}$ can be simplified to $\{(n,m)\}_{n\geq n^{\ast},m=-1,0,1,\ldots ,n-n^{\ast}} $ where  $m=-1$
given that Player 1 has not yet accepted a candidate, otherwise $m$ is the number of candidates seen by Player 2 after moment $n^{\ast}$. Let 
$v^a (n,m)$ be the approximation to the optimal expected reward of Player 2 
after rejecting at moment $n$ the $m$-th candidate to be seen by Player 2 after moment $n^{\ast}$, $m=1,2,\ldots ,N-n^{\ast}$, $n\geq n^{\ast}+m$. The value functions $v^a (n,0)$ and $v^a (n,-1)$ are analogously defined as approximations to the functions $v (n,0)$ and 
$v(n,-1)$, respectively.

For $m\geq 0$, the estimated expected payoff from stopping at the next decision point, 
$w^a (n,m)$ is given by
\begin{equation} 
w^a (n,m)  = \frac{mn}{N(m+1)} \left[ \sum_{k=n+1}^N \frac{1}{k-1} \right] 
+ \frac{n(n-1)}{N(m+1)} \left[ \sum_{k=n+1}^N \frac{1}{(k-1)(k-2)} \right] . \label{app1}
\end{equation}
At a decision point in state $(n,m)$, the estimated reward from accepting the current candidate is $\frac{nm}{N(m+1)}$. Hence, using the near optimal strategy, the current candidate should be accepted if and only if $\frac{nm}{N(m+1)} \geq w^a (n,m)$. 
This leads to the condition
\begin{equation} 
m \geq \lceil S_{2,n}/(1-S_{1,n}) \rceil = K(n), \label{thresh1}
\end{equation} 
where $\lceil x\rceil$ denotes the smallest integer not less than $x$ and $S_{1,n}$ and 
$S_{2,n}$ are given by Equation (\ref{cond11}). Since both $S_{2,n}$ and 
$S_{1,n}$ are decreasing in $n$ and $S_{1,n}<1$ (see the proof of Proposition \ref{prop2}), it follows that 
$K(n)$ is non-increasing in $n$. Thus the near optimal strategy based on this approximation is 
of the following form:
\begin{description}
\item[1.] For $n>n^{\ast}$ the optimal response of Player 2 is to accept the 
$m$-th candidate seen after moment $n^{\ast}$ if and only if $n\geq n_m$, where $n_m$ is 
the smallest $n$ that satisfies $K(n)\leq m$.
\item[2.] The thresholds $n_m$ are non-increasing in $m$. 
\item[3.] For $n\geq n_1$, any candidate should be accepted.  
\item[4.] For finite $N$, there exists a finite $m_0$ such that the $m_0$-th candidate observed by Player 2 after moment $n^{\ast}$ is always accepted. The value of $m_0$ is 
given by the smallest solution $m$ of the inequality $m\geq  \lceil S_{2,n^{\ast}+m}/(1-S_{1,n^{\ast}+m}) \rceil$.
\end{description}
The condition given in Point 4 above states that if $m$ candidates appear in succession at 
moments $n^{\ast}+1,n^{\ast}+2,\ldots $, then Player 2 should reject the first $m_0-1$ of these candidates and accept the $m_0$-th candidate. It follows that $m_0$ can be derived 
without any knowledge of the value functions. 

From the arguments made above, to approximate the expected reward of Player 2 under the corresponding near optimal strategy, it suffices to derive the estimates of the future expected reward starting in the following states: $(n,-1)$ for $0\leq  n\leq N$, and $(n,m)$ for $0\leq m\leq m_0 -1$ and $n^{\ast}+m\leq n\leq N$. Starting from state $(n,m_0-1)$, Player 2 should always accept the next candidate to appear. Hence, $v^a (n,m_0-1)=w^a (n,m_0-1)$, where $w^a (n,m_0-1)$ is given by 
Equation (\ref{app1}). Hence, the required values of the function $v^a (n,m_0 -1)$ can be 
calculated directly. The required values of the function $v^a (n,m_0 -2)$ can then be calculated. For $n\geq n_{m_0-1}-1$, these values can be calculated directly using Equation (\ref{app1}). For $n^{\ast}+m_0-2 \leq n<n_{m_0-1}-1$, these values can be derived by 
considering the distribution of the moment at which the next candidate appears. In general, starting from state $(n,m)$, where $m\geq 0$, this leads to
\begin{eqnarray}
v^a (n,m) & = & \left\{ \sum_{k=n+1}^{n_{m+1}-1} \left[ \frac{mn}{(m+1)k(k-1)}
+ \frac{2n(n-1)}{(m+1)k(k-1)(k-2)} \right] v^a (k,m+1) \right\} + \nonumber \\
&  & + \sum_{k=n_{m+1}}^N \left[ \frac{mn}{(m+1)N(k-1)} + \frac{n(n-1)}{(m+1)N(j-1)(j-2)}\right] . \label{app2}
\end{eqnarray}
It follows that $v^a (n, m_0-2)$ can be calculated for all $n\geq n^{\ast}+m_0-2$, since 
the relevant values of $v^a (k,m_0-1)$ have already been derived. 
Arguing in a similar way, the required values of the function $v^a$ can be derived 
recursively for $m=m_0-3, m_0-4, \ldots ,0$. Finally, 
the function $v^a (n,-1)$ can be calculated using Equation (\ref{opt2p2}) for $n\geq n^{\ast}-1$ or 
(\ref{opt3p2}) for $n<n^{\ast}-1$ and replacing $v(n,j)$ by $v^a (n,j)$ for $j=0,-1$.
These calcuations are illustrated in Example \ref{ex2}.

 The proof of the following theorem is given in the appendix. 

\begin{theorem} \label{thma}
The exact value of the value function, $v(n,p)$ is non-decreasing in $p$ for 
$n>n^{\ast}$ and $p\geq 0$.  In addition, $v(n,-1)<v(n,0)$ for $n^{\ast}<n<N$.
\end{theorem}

The following lemma results from the fact that the near optimal rule based on the number of 
candidates appearing after moment $n^{\ast}$ uses a lower bound approximating 
the probability that the $m$-th such candidate has relative rank 1.
 
\begin{lemma}[to Theorem \ref{thma}] \label{lemma2}
The value function $v^a (n,m)$ gives a lower bound on the corresponding optimal expected future reward 
$v(n,p_{m,n,H})$.
\end{lemma}

The following example derives a near optimal strategy for Player 2 for $N=50$. The optimal response was 
derived in Example \ref{ex1}.

\begin{example} \label{ex2}
First, we derive the values of $K(n)$, from which we can derive $m_0$ and the $n_m$ for 
$m=1,2,\ldots m_0-1$. From Example 1, $n^{\ast}=19$. By direct calculation, we obtain
\begin{eqnarray*}
K(20) = \lceil S_{2,19+1}/(1-S_{1,19+1})\rceil & = & \lceil 8.9334 \rceil = 9 > 1 \\
K(21) = \lceil S_{2,19+2}/(1-S_{1,19+2})\rceil & = & \lceil 4.9930 \rceil = 5 > 2 \\
K(22) = \lceil S_{2,19+3}/(1-S_{1,19+3})\rceil & = & \lceil 3.4392 \rceil = 4 > 3 \\
K(23) = \lceil S_{2,19+4}/(1-S_{1,19+4})\rceil & = & \lceil 2.6040 \rceil = 3 \leq 4. 
\end{eqnarray*}
It follows that Player 2 should always accept the fourth candidate to appear after moment 
19. Continuing these calculations
\begin{eqnarray*}
K(24) = \lceil S_{2,24}/(1-S_{1,24})\rceil & = & \lceil 2.0801 \rceil = 3 \\
K(25) = \lceil S_{2,25}/(1-S_{1,25})\rceil & = & \lceil 1.7193 \rceil = 2 \\
K(28) = \lceil S_{2,28}/(1-S_{1,28})\rceil & = & \lceil 1.0891 \rceil = 2 \\
K(29) = \lceil S_{2,29}/(1-S_{1,29})\rceil & = & \lceil 0.9567 \rceil = 1. \\
\end{eqnarray*}
It follows that the first, second and third candidates to appear after moment 19 should be 
accepted when $n\geq 29$, $n\geq 25$ and $n\geq 23$, respectively.

Now we derive the approximate future expected reward for this near optimal strategy. Starting from the state $(n,3)$, where $n\geq n^{\ast}+3=22$, the next candidate to appear should always be accepted by Player 2. It follows from Equation (\ref{app1}) that
\begin{equation}
v^a (n,3) = \frac{3n}{200} \left[ \sum_{k=n+1}^{50} \frac{1}{k-1} \right] 
+ \frac{n(n-1)}{200} \left[ \sum_{k=n+1}^{50} \frac{1}{(k-1)(k-2)} \right] .
\label{exam21}
\end{equation}
These values can be calculated directly.

Analagously, starting from the state $(n,2)$, where $n\geq n_3 -1=22$, the next candidate  to appear should be accepted by Player 2. It follows that
\begin{equation}
v^a (n,2) = \frac{2n}{150} \left[ \sum_{k=n+1}^{50} \frac{1}{k-1} \right] 
+ \frac{n(n-1)}{150} \left[ \sum_{k=n+1}^{50} \frac{1}{(k-1)(k-2)} \right] .
\label{exam22}
\end{equation}
Again, these values can be calculated directly.

Starting from the state $(n,2)$, where $21=n^{\ast}+2\leq n<n_3 -1=22$ (i.e. $n=21$), 
from Equation (\ref{app2}), we obtain
\begin{equation}
v^a (21,2) = \frac{2v^a (22,3)}{33} + \frac{21}{75} \left[ \sum_{k=23}^{50} \frac{1}{k-1} 
\right] + \frac{7}{5} \left[ \sum_{k=23}^{50} \frac{1}{(k-1)(k-2)}\right] .\label{exam23}
\end{equation}
This can be calculated, since $v^a (22,3)$ has already been derived.

Similarly, starting from the state $(n,1)$, where $n\geq n_2 -1=24$, the next candidate  to appear should be accepted by Player 2. It follows that
\begin{equation}
v^a (n,1) = \frac{n}{100} \left[ \sum_{k=n+1}^{50} \frac{1}{k-1} \right] 
+ \frac{n(n-1)}{100} \left[ \sum_{k=n+1}^{50} \frac{1}{(k-1)(k-2)} \right] .
\label{exam24}
\end{equation}
These values can be calculated directly.

Starting from the state $(n,1)$, where $20=n^{\ast}+1\leq n<n_2 -1=24$, 
from Equation (\ref{app2}), we obtain
\begin{eqnarray}
v^a (n,1) & = & \frac{n}{2} \left[ \sum_{k=n+1}^{24} \frac{v^a (k,2)}{k(k-1)} 
\right] + n(n-1) \left[ \sum_{k=n+1}^{24} \frac{v^a (k,2)}{k(k-1)(k-2)} \right] + 
\nonumber \\
& & + \frac{n}{100} \left[ \sum_{k=25}^{50} \frac{1}{k-1} \right] + \frac{n(n-1)}{100}
\left[ \sum_{k=25}^{50} \frac{1}{(k-1)(k-2)} \right] . \label{exam25}
\end{eqnarray}
These values can be derived, since the appropriate values of $v^a (k,2)$ have 
already been calculated.

Starting from the state $(n,0)$, where $n\geq n_1 -1=28$, the next candidate  to appear should be accepted by Player 2. It follows that 
\begin{equation}
v^a (n,0) = \frac{n(n-1)}{100} \left[ \sum_{k=n+1}^{50} \frac{1}{(k-1)(k-2)} \right] .
\label{exam26}
\end{equation}
These values can be calculated directly.

For $19=n^{\ast}\leq n<n_1 -1 =28$, from Equation (\ref{app2}), we obtain
\begin{equation}
v^a (n,0) = 2n(n-1) \left[ \sum_{k=n+1}^{28} \frac{v^a (k,1)}{k(k-1)(k-2)} \right]
+ \frac{n(n-1)}{50} \left[ \sum_{k=29}^{50} \frac{1}{(k-1)(k-2)} \right] . \label{exam27}
\end{equation}
These values can be derived, since the appropriate values of $v^a (k,1)$ have 
already been calculated.

Considering the distribution of the moment at which Player 1 accepts a candidate [using the appropriate adaptation of Equation (\ref{exam6})], for $n\geq 18$, we obtain
\begin{equation}
v^a (n,-1) = \sum_{k=n+1}^{50} \frac{nv^a (k,0)}{k(k-1)} . \label{exam28}
\end{equation}
These values can be derived, since the appropriate values of $v^a (k,0)$ have 
already been calculated.

Thus all the necessary approximations of the future expected reward, $v^a (n,m)$, can be derived for $n\geq n^{\ast}-1$ and $m=-1,0,1,2,3$ 
using Equations (\ref{exam21})-(\ref{exam28}). Direct calculation leads to 
$v^a (18,-1)\approx 0.145868$. 

When $n< n^{\ast}$, Player 2 should accept a candidate if and only if $n\geq n_0^a$, where $n_0^a$ is the smallest value of $n$ satisfying $\frac{n_0^a}{N}\geq v^a (n,-1)$. 
For $n_0^a -1 \leq n<n^{\ast}$, using the appropriate adaptation of Equation
 (\ref{opt6p2}), we obtain
\begin{equation}
v^a (n,-1) = \frac{n}{50} \left[ \sum_{k=n+1}^{18} \frac{1}{k-1} \right] 
+ \frac{n}{18} v^a (18,-1). \label{exam26a} 
\end{equation}
Using a near optimal rule based on this approximate future expected reward leads to the condition that $n_0^a$ is the smallest integer value of $n$ to satisfy the following inequality
\begin{equation}
1 \geq \frac{Nv^a (n^{\ast}-1,-1)}{n^{\ast}-1} + \sum_{k=n+1}^{n^{\ast}-1} \frac{1}{k-1}. \label{bounda}
\end{equation}
It follows from direct calculations that Player 2 should accept a candidate when $11=n_0^a \leq n < n^{\ast}$ and $v^a (0,-1)=v^a (10,-1)\approx 0.203155$.
To summarise, the approximately optimal strategy derived in this way is as follows:
\begin{description}
\item[i)] accept a candidate if $11\leq n\leq 18$, 
\item[ii)] accept the first candidate seen after moment 19, if and only if $n\geq 29$,
\item[iii)] accept the second candidate seen after moment 19, if $n\geq 25$.
\item[iv)] accept the third candidate seen after moment 19, if and only if $n\geq 23$,
\item[v)] always accept the fourth candidate seen after moment 19.
\end{description}
\end{example}
Comparing this near-optimal strategy with the optimal strategy, the only difference occurs when the second candidate to appear after moment 19 appears at moment 24. The simplified rule always rejects such a candidate. The optimal strategy accepts such a candidate if and only if the first candidate was seen before moment 22. The approximate value of the near 
optimal rule differs from the optimal value in the sixth position after the decimal point. It should be noted that this approximation is a lower bound on the expected reward obtained using the near optimal solution, since the probability of a candidate having relative rank 1 is underestimated. Hence, using the simplified strategy Player 2 has an expected 
payoff of at most 0.001\% lower than the expected reward obtained under the optimal strategy. For the cases considered in Table \ref{tab1}, i.e. $3\leq N\leq 10$ and $N=20$), the simplified strategy was identical to the optimal strategy. 

This simplified strategy is considered in the following section, which gives an number of results regarding the asymptotic form of the optimal response of Player 2 and the value of the game to him when $N\rightarrow \infty$. 

\section{Asymptotics of the Optimal Response of Player 2} \label{sec5}

Suppose $N\rightarrow \infty$ and let $t$, referred to as the time, be the proportion of objects already observed. Gilbert and Mosteller (1966) showed that Player 1 should reject 
any object appearing before time $t^{\ast}=e^{-1}$ and accept the first candidate to 
appear after time $t^{\ast}$. The value of the game to Player 1 is $e^{-1}$. 

Now we consider the optimal response of Player 2. Analogously to the problem for finite $N$,
let $v(t,-1)$ denote the future expected reward of Player 2 when he is still searching at time $t$ and Player 1 has not yet accepted an object. Similarly, $v(t,0)$ denotes the future expected reward of Player 2  when Player 1 has already accepted an object, Player 2 is still searching, but has not observed a candidate since time $e^{-1}$. Finally, for $t\geq e^{-1}$, $v(t,p)$ denotes the future expected reward of Player 2 immediately after rejecting a  candidate when the probability that this candidate has relative rank 1 is $p$. The state at the 
corresponding decision point will be denoted by $(t,p)$.

Using an analogous argument to the one used for finite $N$, the set of the appearance times 
of the candidates observed by Player 2 after time $e^{-1}$ is a rich enough history to 
calculate the probability that the current candidate has relative rank 1. 
Let $p_m$ denote the probability that the $m$-th candidate seen by Player 2 after time $e^{-1}$ has relative rank 1 given the times at which such candidates have appeared, $(\mu_1 =s_1,\mu_2 =s_2, \ldots ,\mu_m =s_m)$. These probabilities can be calculated 
inductively by setting $p_1 =0.5$ and using 
\begin{equation}
p_{m+1} = \frac{s_{m+1} p_m + s_m (1-p_m)}{s_{m+1}p_m + 2s_m (1-p_m)} .
\label{asymp3}
\end{equation}
It should be noted that this equation can be obtained as the limit of Equation (\ref{induc1})
by setting $s_m = \frac{n}{N}$ and $s_{m+1} = \frac{k}{N}$. Using an argument analogous to the one used for finite $N$, it can be shown that for $m\geq 2$, $p_m > \frac{m}{m+1}$. 

Assume that $t\geq e^{-1}$ and Player 1 has not yet accepted a candidate. Considering the distribution of the time at which the next object of relative rank 1 appears, i.e. transitions between the states $(t,-1)$ and $(s,0)$, it follows that 
\begin{equation}
v(t,-1) = \int_t^1 \frac{tv(s,0)ds}{s^2}. \label{asymp1}
\end{equation}
After Player 1 has accepted an object, then the next object of relative rank 1 or 2 to appear is the first candidate to be observed by Player 2 after time $e^{-1}$. The probability that 
such an object has relative rank one is 1/2. 
Considering the distribution of the time at which such a candidate appears, it follows that 
\begin{equation}
v(t,0) = \int_t^1 \frac{2t^2v(s,1/2)ds}{s^3}. \label{asymp2}
\end{equation}

Define $w(t, p)$ to be the expected future reward gained by Player 2 immediately after he has rejected a candidate in state $(t,p)$ when he accepts the next candidate to appear. 
Using the extension of the law of total probability in conjunction with the 
expected future reward from such a strategy given the relative rank of the most recent candidate [see Equations (\ref{tm6}) and (\ref{yt})], we obtain
\begin{equation}
w(t,p) = pz(t)+(1-p)y(t)= \int_t^1 \left( \frac{pt}{s} + \frac{(1-p)t^2}{s^2} \right) ds =t[ (1-p)(1-t) -p\ln t].
\label{asymp4}
\end{equation} 
The probability that the current candidate has absolute rank 1 is $pt$. Hence,
using an OSLA rule, a candidate should be accepted in state $(t,p)$ if and only if 
\begin{equation}
p \geq q(t) = \frac{1-t}{2-t +\ln t}. \label{asymp5}
\end{equation}
Note that 
\[
q'(t) = -\frac{\ln t +1/t}{(2-t+\ln t)^2} <0  \mbox{ for } t\in [e^{-1},1] .
\]
It follows that if Player 2 should stop in state $(t,p_m)$ according to such an OSLA rule, then 
Player 2 should stop at the next decision point $(s,p_{m+1})$, since $p_{m+1}>p_m$ and 
$q(s)<q(t)$. Hence, this OSLA rule is the optimal strategy for $t\geq e^{-1}$. 

In particular, the first candidate to be seen by Player 2 after time $e^{-1}$ should be 
accepted if and only if
\begin{equation}
\frac{1}{2} \geq \frac{1-t}{2-t +\ln t} \Rightarrow t\geq -\ln t \Rightarrow t\geq e^{-t}. \label{asymp6}
\end{equation}
It follows that the first candidate to be observed by Player 2 after time $e^{-1}$ should be 
accepted if and only if $t\geq t_1$, where $t_1 \approx 0.567143$ is the solution of the equation $t=e^{-t}$. This solution may be found, for example, using value iteration.

However, derivation of the exact form of the optimal response and
value of such a game to Player 2 is highly complex, since it is necessary to integrate over the 
possible trajectories of the process describing the transitions of the state of Player 2. In the following two subsections, upper and lower bounds are derived for this value. 

\subsection{An upper bound on the value of the game to Player 2} \label{upper}

One bound can be found very simply by considering the value to Player 2 of the analogous game in 
which the players simultaneously observe each object and Player 1 always has priority. In such a case, Player 2 can compare the value of candidates that appear later with the value 
of the object chosen by Player 1. The value of such a game to Player 2 is $e^{-1.5}\approx 
0.223130$ (see Sakaguchi 1980). 

A tighter upper bound can be found by considering the following slight adaptation of the game. Suppose that after rejecting a candidate at time 
$t$, where $t\geq e^{-1}$, Player 2 is told the relative rank of this candidate (either one or two). Define $u(t, -1)$ to be the optimal future expected reward of Player 2 in such a game when Player 1 has not yet accepted a candidate. Let $u(t, 0)$ be the optimal future expected reward of Player 2 when Player 1 has accepted a candidate, but Player 2 has not observed a candidate since then. Finally, let $u(t, i)$, $i=1,2$, $t\geq e^{-1}$ be the 
optimal future expected reward of Player 2 immediately after rejecting a candidate of relative rank $i$ at time $t$. 

Given that the previous candidate had relative rank one, the next candidate to appear must also have relative rank one. Thus given this information, the problem faced by Player 2 reduces to the standard secretary problem. Hence, in this case Player 2 should accept the next candidate to appear and $u(t, 1)=-t\ln t$. It can be seen that 
$u(t, 0)=u(t, 2)$, since in both cases the next candidate to appear will have either relative rank one or two, each with a probability of 1/2. In this case, the next candidate should be 
accepted if and only if $t\geq t_1 \approx 0.567143$. It follows that for $t\geq t_1$
\begin{equation}
u(t, 2) = w(t, 0) = \frac{z(t)+y(t)}{2} = \frac{-t\ln t + t(1-t)}{2}. \label{upper1}
\end{equation}
For $e^{-1}\leq t< t_1$, by considering the distribution of the time at which the next 
object of relative rank 1 or 2 appears, we obtain
\begin{equation}
u(t, 2) = \int_t^{t_1} \frac{t^2 [u(s, 1)+u(s, 2)]ds}{s^3} + \int_{t_1}^1 \frac{t^2ds}{s^2}.  \label{upper2}
\end{equation}
Dividing this equation by $t^2$ and then differentiating, we obtain a differential equation for 
$u(t,2)$. From the boundary condition at $t=t_1$, obtained by substituting $t=t_1$ into 
Equation (\ref{upper1}), it follows that
\begin{equation}
u(t, 2) = c_1 t + \frac{t\ln ^2 t}{2},  \label{upper3}
\end{equation}
where $c_1 = 1-t_1-0.5t_1^2\approx 0.272031$. For $t\geq e^{-1}$, considering the distribution of the time at which Player 1 accepts an object, we obtain
\begin{equation}
u(t, -1) = \int_t^1 \frac{tu(s, 2)ds}{s^2} . \label{upper4}
\end{equation}
It follows by direct integration using the form of $u(t, 2)$ that 
\begin{equation}
u(t,-1) = \left\{ \begin{array}{cc} 
c_2 t -c_1 t\ln t - \frac{t\ln ^3 t}{6}, & e^{-1}\leq t<t_1 ,  \\
t[t-1-\ln t], & t\geq t_1 , \end{array} \right.\label{upper5}
\end{equation}
where $c_2 = t_1 +t_1^2+\frac{t_1^3}{3}-1 \approx -0.050398$.

Suppose that Player 2 accepts candidates in the interval $[t_0^u ,e^{-1})$. Since Player 2 
sees all the objects up to time $e^{-1}$, such a candidate must have relative rank one. 
Considering the distribution of the time at which the first object of relative rank one appears after time $t$, it follows that
\begin{equation}
u(t, -1) = \int_t^{e^{-1}} \frac{tds}{s} + \int_{e^{-1}}^1 \frac{u(s,2)ds}{s^2}.
\label{upper6}
\end{equation}
From the border condition at $t=e^{-1}$, this leads to $u(t, -1)=c_3 t-t\ln t$  for $t\in [t_0^u ,e^{-1}]$, where $c_3 = \frac{3t_1^2 +2_1^3-5}{6}\approx
-0.611700$. Player 2 should accept a candidate on this interval if and 
only if $t\geq t_0^u$, where $t_0^u$ satisfies the equation $t_0^u = u(t_0^u ,-1)$. This gives 
$t_0^u = \exp (c_3 -1)\approx 0.199548$. 

For $t<t_0^u$, $u_1 (t_0^u, -1)=t_0^u$. Hence, it can be seen that for the original game, the value 
of the game to Player 2, $v(0,-1)$, satisfies $v(0,-1)<t_0^u\approx 0.199548$. 

\subsection{A lower bound on the value of the game to Player 2} \label{lower}

As previously, we consider a near optimal strategy based on the number of candidates that 
Player 2 has seen since time $e^{-1}$. A lower bound on the probability that the 
$m$-th such candidate has relative rank 1 is given by $\frac{m}{m+1}$. Suppose Player 2
uses the appropriate near optimal strategy based on this approximation. The state of Player 2 on observing the $m$-th candidate after time $e^{-1}$ at time $t$ is defined to be 
$(t,m)$. Define $v^{\infty}(t, m)$ to be the corresponding approximation of the value of 
the game to Player 2 from future search immediately after rejecting a candidate in this state.
The index $\infty$ denotes the fact that Player 2 is assumed to have a perfect memory and 
knows exactly how many candidates have appeared since time $e^{-1}$.  The function
$v^{\infty}$ is an underestimate of the future expected reward of Player 2 (this can be seen by adapting the proof given in the appendix to the case where $N\rightarrow \infty$). Analogously, let $v^{\infty}(t,-1)$ be the expected future reward of Player 2 when neither 
player has accepted an object at time $t$. Additionally, for $t\geq e^{-1}$ define $v^{\infty}(t,0)$ to be the expected future reward of Player 2 when Player 1 has already accepted an object, but Player 2 has not seen a candidate since time $e^{-1}$. Finally, let 
$w^{a}(t,m)$ for $m\geq 0$, $t\geq e^{-1}$ be the approximation of the 
expected reward of Player 2 from accepting the next candidate to appear when starting from state $(t,m)$. Using the law of total probability,
\begin{equation} 
w^a (t,m) = \frac{mz(t)+y(t)}{m+1} = \frac{t(-m \ln t + 1-t)}{m+1}. \label{infapp1}
\end{equation}
Given that Player 2 uses the corresponding near optimal policy, he should accept a candidate in state $(t,m)$ 
if and only if $w^a (t,m) \leq \frac{mt}{m+1}$. It follows that Player 2 should accept a 
candidate in state $(t,m)$, if and only if $t\geq t_m$, where $t_m$ satisfies 
\begin{equation}
-\ln t_m = \frac{m-1+t_m}{m} \Rightarrow t_m = \exp \left[ \frac{1-m-t_m}{m}\right]. \label{boundcon1}
\end{equation}

 This equation can be solved using value iteration. Table \ref{tab3} gives the value of the
time thresholds $t_m$ to four decimal places for chosen values of $m$

\begin{table}[h]
\caption{Asymptotic time thresholds based on the number of candidates to appear}\label{tab3}%
\begin{tabular}{@{}llll@{}}
\toprule
$m$ & $t_m$ & $m$ & $t_m$ \\
\midrule
1 & 0.5671 & 20 & 0.3795 \\
2 & 0.4777 & 50 & 0.3725 \\
3 & 0.4429 & 100 & 0.3702 \\
4 & 0.4248 & 200 & 0.3690 \\
5 & 0.4137 & 500 & 0.3683 \\
6 & 0.4062 & 1 000 & 0.3681 \\
7 & 0.4008 & 10 000 & 0.3679 \\
8 & 0.3967 & 100 000 & 0.3679 \\
9 & 0.3935 & 1 000 000 & 0.3679 \\
10 & 0.3910 & $\infty$ & $e^{-1}\approx 0.3679$ \\
\botrule
\end{tabular}
\footnotetext{Source: Author's calculations using value iteration}
\end{table}
The case $m=\infty$ corresponds to the classical secretary problem, as in this case a 
candidate has relative rank 1 with probability 1. Note that for $t\geq t_{m+1}$,
$v^{\infty} (t,m)=w^a (t,m)$. Also, for $m\geq 0$ and $e^{-1}\leq t\leq t_{m+1}$, 
considering the distribution of the time at which the next candidate appears to Player 2, we obtain
\begin{equation}
v^{\infty} (t,m) = \int_{t}^{t_{m+1}} \frac{v^{\infty}(s,m+1) (mst+2t^2)ds}{s^3 (m+1)}
+ \int_{t_{m+1}}^1 \frac{(mst+t^2)ds}{s^2 (m+1)}. \label{lower1}
\end{equation} 
Similarly, for $t\geq e^{-1}$,
\begin{equation}
v^{\infty} (t,-1) = \int_t^1 \frac{tv^{\infty} (s,0)ds}{s^2}. \label{lower2}
\end{equation}
Suppose that Player 2 accepts a candidate on the interval $t_0^{\infty} \leq t<e^{-1}$. 
Considering the distribution of the time at which the next candidate appears, we obtain
\begin{equation}
v^{\infty} (t,-1) = \int_t^{e^{-1}} \frac{tds}{s} + \int_{e^{-1}}^1 \frac{tv^{\infty} (s,0)ds}{s^2} \label{lower3},
\end{equation}
where $t_0^{\infty}$ satisfies $t_0^{\infty} = v^{\infty}(t_0^{\infty})$. For $t\leq 
t_0^{\infty}$, $v^{\infty}(t,-1)=t_0^{\infty}$ and hence $t_0^{\infty}$ gives a lower 
bound on the value of the game to Player 2.

In order to find a reasonable approximation to this expected reward, it is assumed that Player 2 has a limited memory, namely that he can 
remember up to $k$ previous candidates, $k\in \{ 0,1,2,\ldots \}$. Under this assumption, the $m$-th candidate to appear to Player 2 after time $e^{-1}$ is treated as the 
$k+1$-th such candidate whenever $m\geq k+1$. The state at such a decision point is defined to be $(t,k+1)$. After rejecting such a candidate, Player 2
immediately forgets the first of the previous candidates. 

The corresponding approximations of the future expected rewards will be denoted by 
$v^k (t,m)$, $m=-1,0,1,\ldots ,k$. In this case, the estimate of the probability that a candidate has relative rank 1 is less than or equal to the corresponding estimate in the case where Player 2 has an unlimited (or simply larger) memory. It follows that 
$v^k (t,m)< v^{k+1} (t,m)<v^{\infty}(t,m)<v(t,p_{m,t,H})$ for $t\in [0,1]$ and $m\leq k$. 

Assume that for a given $k$ Player 2 uses the time thresholds $t_1,t_2,\ldots t_{k+1}$ (see 
Table \ref{tab3}) such that for $m\leq k$ the $m$-th candidate to appear to Player 2 after time $e^{-1}$ is accepted at time $t$ if and only if $t\geq t_m$ and for $m>k$ the 
$m$-th candidate to appear is accepted at time $t$ if and only if $t\geq t_{k+1}$.

Using an OSLA rule based on such an approximation of the future expected reward, when starting from state $(t,m)$, where $t\geq t_{m+1}$, then Player 2 will accept the next candidate to appear. Hence, for 
$t\geq t_{m+1}$ and $0\leq m\leq k$, $v^k (t,m) = w^a (t,m)$, where $w^a (t,m)$ is 
given by Equation (\ref{infapp1}). 

 Note that for $e^{-1}\leq t< t_{k+1}$, $v^k (t,k)=v^k (t_{k+1},k)$. This results from the facts that
Player 2 will reject any candidate appearing before time $t_{k+1}$ and the appearance of such a candidate will not change the state of Player 2's memory. It follows from Equation (\ref{infapp1})  that for $e^{-1}\leq t< t_{k+1}$
\[
v^k (t,k) = v^k (t_{k+1},k) = \frac{t_{k+1} (-k\ln t_{k+1} + 1 -t_{k+1})}{k+1}.
\]
From the boundary condition given by Equation (\ref{boundcon1}), we obtain
\begin{equation}
v^k (t,k) = \frac{(k^2 +k+1)t_{k+1} - t^2_{k+1}}{(k+1)^2}. \label{lower4}
\end{equation}
Hence, for $t\geq e^{-1}$, the value function $v^k (t,k)$ can be calculated directly. 

For $m=0,1,\ldots k-1$ and $e^{-1}\leq t<t_{m+1}$, the value function 
$v^k (t,m)$ satisfies the equation 
\begin{equation}
v^{k} (t,m) = \int_{t}^{t_{m+1}} \frac{v^{k}(s,m+1) (mst+2t^2)ds}{s^3 (m+1)}
+ \int_{t_{m+1}}^1 \frac{(mst+t^2)ds}{s^2 (m+1)}. \label{lower5}
\end{equation} 
Note that this equation is analogous to Equation (\ref{lower1}) with $v^{\infty} (t,m)$ 
and $v^{\infty} (t,m+1)$ being replaced by their approximations $v^{k} (t,m)$ 
and $v^{k} (t,m+1)$, respectively. Analogously, for $t\geq e^{-1}$, the value function $v^k (t,-1)$ satisfies the equation
\begin{equation}
v^k (t, -1) = \int_t^1 \frac{tv^k (s,0)ds}{s^2}. \label{lower6}
\end{equation}
Since the functional form of $v^k (s,0)$ changes, the simplest way to calculate this function for $t\in [t_{i+1},t_i )$, where $i=1,2,\ldots ,k$, is to condition on whether Player 1 observes (and thus accepts) a candidate before time 
$t_i$ or not. Let $A(t,t_i)$ be the event that Player 1 does not accept a candidate in the interval $(t,t_i )$. This is equivalent to the event that the best object to appear before time $t$ is the best object to appear 
before time $t_i$, thus $P[A(t,t_i)] = \frac{t}{t_i}$. Given that $A(t,t_i)$ occurs the future expected reward of Player 2 is $v^k (t_i ,-1)$. Using the law of total probability, we obtain that for  $t\in [t_{i+1},t_i )$
\begin{equation}
v^k (t, -1) = \frac{tv^k (t_i ,-1)}{t_i} + \int_t^{t_i} \frac{tv^k (s,0)ds}{s^2}. 
\label{lower7}
\end{equation}
It should be noted that for $t\geq t_1$, Equation (\ref{lower6}) can be interpreted as a particular case of Equation (\ref{lower7}), since by definition $v^k (1,-1)=0$. Also, a similar 
approach can be used to solve Equation (\ref{lower5}) when $e^{-1}\leq t<t_{m+2}$. This will be illustrated when deriving the value function $v^1 (t,0)$ for $e^{-1}\leq t<t_{2}$
in Example \ref{ex3}.

The resulting system of equations for $v^k (t,m)$, $m=-1,0,1,\ldots k-1$ and $t\geq e^{-1}$ can be solved recursively using the following procedure:

\begin{description}
\item[1.] The interval $[e^{-1},1]$ is split into intervals $[e^{-1},t_{k+1}), [t_{k+1},t_{k}), 
[t_k, t_{k-1}),\ldots ,[t_1 ,1]$. 
\item[2.] To initiate the recursion procedure, on the interval $[t_1 ,1]$, $v^k (k,m)=w^a (t,m)$, for $m=0,1,\ldots ,k$, where $w^a (t,m)$ is 
given by Equation (\ref{infapp1}). 
\item[3.] In step $i$, $1\leq i\leq k$, of the procedure, we derive the value functions, $v^k (t,m)$ on the interval $[t_{i+1},t_i]$ using the boundary conditions at $t_i$ resulting from the continuity of the value functions. For $m=i,i+1,\ldots k$, $v^k (k,m)=w^a (t,m)$. Hence, $v^k (k,i-1)$ can be derived directly from Equation (\ref{lower5}). The value 
functions $v_k (k,m), m=i-2,i-3,\ldots ,0$ can then be derived recursively using 
the same equation. 
\item[4.] In the $k+1$-th (final) step of the procedure, we calculate the value functions on the interval 
$[e^{-1},t_{k+1})$. On this interval, the function $v^k (t,k)$ is given by Equation (\ref{lower4}). The functions 
$v^k (t, m), m=k-1,k-2,\ldots ,0$ can be derived recursively using Equation (\ref{lower5}).
Finally, the value function $v^k (t,-1)$ can be derived recursively on the intervals $[t_1 ,1], [t_2 ,t_1], [t_3 ,t_2], \ldots , [t_{k+1},t_k], [e^{-1},t_{k+1})$ using 
Equation (\ref{lower7}).
\end{description}

Once this set of value functions has been derived for $t\geq e^{-1}$, since Player 1 rejects all objects appearing before time $e^{-1}$, it remains to derive $v^k (t,-1)$ for 
$t^k_0 \leq t<e^{-1}$, where $t^k_0$ is the earliest time at which Player 2 is willing to accept a candidate. Note that any candidate appearing to Player 2 before time $e^{-1}$ must have relative rank 1 and from the form of the near optimal strategy $t^k_0$ satisfies $t^k_0 = v^k (t^k_0,-1)$. Arguing as in the derivation of Equation (\ref{lower7}), we obtain
\begin{equation}
v^k (t,-1) = etv^k (e^{-1},-1)+\int_t^{e^{-1}} \frac{tds}{s}.  \label{lower8}
\end{equation}
For $0\leq t< t^k_0$, $v^k (t,-1)=v^k (t^k_0,-1)=t^k_0$.

\begin{example} \label{ex3}
Assume that $k=1$. When $t\geq t_1 \approx 0.5671$, the next candidate to appear to Player 2 is always accepted. From Equations (\ref{infapp1}) and (\ref{lower6}), it follows that
\begin{eqnarray}
v^1 (t,m) & = & w^a (t,m) = \frac{t(-m\ln t+1-t)}{m+1}, \hspace{.2in} m\in \{ 0,1\} \label{k11} \\
v^1 (t,-1) & = & \int_t^1 \frac{tv^1 (s,0)ds}{s^2} = t(t-1-\ln t) . \label{k12}
\end{eqnarray}

Now consider $t_2 \approx 0.4777 \leq t<t_1 \approx 0.5671$. If Player 2 has already observed a candidate after 
time $e^{-1}$, then the next candidate to appear will always be accepted. If Player 2 has not yet observed a 
candidate after time $e^{-1}$, then the next candidate to appear to Player 2 will only be accepted if it appears after time $t_1$. Hence, from Equations (\ref{infapp1}), (\ref{lower5}) and (\ref{lower7}), it follows that
\begin{eqnarray}
v^1 (t,1) & = & w^a (t,1) = \frac{t(-\ln t+1-t)}{2}, \label{k21} \\
v^1 (t,0) & = &  \int_{t}^{t_1} \frac{2t^2 v^{1}(s,1) ds}{s^3} 
+ \int_{t_1}^1 \frac{t^2 ds}{s^2}  \label{k22}, \\
v^1 (t,-1) & = & \frac{tv^1 (t_1 ,-1)}{t_1} + \int_t^{t_1} \frac{tv^1 (s,0)ds}{s^2} . \label{k23}
\end{eqnarray}
Using the expression for $v^1 (s,1)$ from Equation (\ref{k21}), Equation (\ref{k22}) leads to
\begin{equation}
v^1 (t,0) = t^2 \left( \ln t -1-\ln t_1 + \frac{1+\ln t_1}{t_1} \right) -t\ln t \label{k22part2}.
\end{equation}
From the boundary condition for $t_1$, $-\ln t_1 =t_1$, it follows that
\begin{equation}
v^1 (t,0) = t^2\ln t-t \ln t +c_1 t^2 \label{k22part3},
\end{equation}
where $c_1 = \frac{1}{t_1}+t_1-2\approx 0.330366$. It should be noted that the values 
of the constants $c_i$ given here are specific to this example.

Having derived $v^1 (t,0)$ for $t\in [t_2 ,t_1)$, we can now derive $v^1 (t,-1)$ on this interval using Equation (\ref{k23}), together with the boundary condition $-\ln t_1 = t_1$. 
This leads to the equation
\begin{equation}
v^1 (t,-1) = \frac{t \ln^2 t}{2}-t^2 \ln t +(1-c_1)t^2 +c_2 t \label{k23part2},
\end{equation}
where $c_2 = -t_1 - \frac{t_1^2}{2}\approx -0.727969$. 

For $t\in [e^{-1},t_2)$, from Equation (\ref{lower4}) it follows that 
\begin{equation}
v^1 (t,1) = c_3, \mbox{ where } c_3 = \frac{3t_2 - t_2^2}{4} \approx 0.301210 \label{k31}.
\end{equation} 
Starting in the state $(t,0)$ for $e^{-1}<t<t_2$, the next candidate
to appear to Player 2 is the next object of relative rank 1 or 2. No such candidate appears before time $t_2$ with probabilty $\frac{t^2}{t_2^2}$. Hence, using the 
law of total probability, we obtain
\begin{equation}
v^1 (t,0) = \int_t^{t_2} \frac{2t^2 c_3 ds}{s^3} + \frac{t^2 v^1 (t_2 ,0)}{t_2^2} \label{k32} .
\end{equation}
Using the boundary condition 
$-\ln t_2 = (1+t_2)/2$, this leads to 
\begin{equation}
v^1 (t,0) = c_3 +c_4 t^2, \label{k32part2}
\end{equation}
where $c_4 = \frac{1-2t_2}{4}-\frac{1}{4t_2}+c_1\approx -0.181843$.

To derive $v^1 (t,-1)$ for $e^{-1}<t<t_2$, from Equation (\ref{lower7}), we obtain
\begin{equation}
v^1 (t, -1) = \frac{tv^1 (t_2 ,-1)}{t_2} + \int_t^{t_2} \frac{tv^1 (s,0)ds}{s^2}. 
\label{k33}
\end{equation}
In conjunction with the boundary condition $-\ln t_2 = (1+t_2)/2$, this leads to 
\begin{equation}
v^1 (t,-1) = c_3 -c_4 t^2 +c_5 t, \label{k33part2}
\end{equation}
where $c_5 = c_2 + \frac{5(t_2^2-1)}{8} + t_2 (2-c_1 + c_4) \approx -0.499690$. 

Since we have derived all of the value functions for $t\geq e^{-1}$, it remains to derive 
$v^1 (t,-1)$ for $t^1_0 \leq t<e^{-1}$, where $t_0^1$ is the first time at which Player 2 is willing to accept a candidate.  From Equation (\ref{lower8}), we obtain
\begin{equation}
v^1 (t,-1) = etv^1 (e^{-1},-1)+\int_t^{e^{-1}} \frac{tds}{s}.  \label{k41}
\end{equation}
Together with the boundary condition $v^1 (e^{-1},-1)=c_3-c_4e^{-2}+c_5e^{-1}$, we obtain 
\begin{equation} 
v^1 (t,-1) = -t\ln t+c_6 t,  \label{k41part2}
\end{equation}
where $c_6 = c_5 -1-c_4e^{-1}+ec_3 \approx -0.614019$. 

Based on the corresponding near optimal strategy, Player 2 should accept a candidate as long as 
$v^1 (t,-1) \geq t$. It follows that Player 2 becomes willing to accept a candidate at time 
$t^1_0$, where $t^1_0$ satisfies
\begin{equation}
t^1_0 = -t^1_0 \ln t^1_0 + c_6 t^1_0 \Rightarrow \ln t^1_0 = -1+c_6 \Rightarrow t^1_0 = e^{-1+c_6} \approx 0.199086.
\end{equation}

For $0\leq t< t^1_0$, $v^1 (t,-1) = e^{-1+c_6}$. It should be noted that this is a lower 
bound on the value of the game to Player 2.
\end{example}

A general upper bound on the underestimation of the value function due purely to the 
limited memory of Player 2 (and not the underestimation of the probability of a candidate having relative rank 1 due to just counting the number of candidates after time $e^{-1}$) can be derived as follows: Suppose Player A and Player B both play the role of Player 2 in 
two identical realisations of the game. It is assumed that Player A can count exactly the number of candidates that have appeared after time $e^{-1}$, while Player $B$ can count up to $k$. That is to say, Player A uses the near optimal rule based on the time thresholds 
$\{ t_m \}_{m=1}^{\infty}$ and Player B uses the near optimal rule based on the time thresholds 
$\{ t_m \}_{m=1}^{k+1}$. It should be noted that the actions of the players can only differ on the interval $(e^{-1},t_{k+1})$ when at least $k+2$ candidates according to Player 2 appear in this interval. If the players' decisions differ, then Player $A$ accepts a candidate that Player $B$ rejects. Since every such candidate has a relative rank $\leq 2$, the probability of such an event occurring is bounded above by the probability of $k+2$ objects of rank 1 or 2 appearing in the interval $(e^{-1},t_{k+1})$. The number of objects of relative rank 1 or 2 in this interval has a Poisson distribution with parameter $\lambda_k = 2(1+\ln t_{k+1})$ 
(see Resnick, 2008). From the boundary condition for $t_{k+1}$, see Equation (\ref{boundcon1}), it follows that $\lambda_k = \frac{2(1-t_{k+1})}{k+1}$.  Suppose that both players are still searching at time $e^{-1}$. In this case, let $C$ be the event that Player A accepts a candidate that Player B rejects.  It follows that
\begin{equation}
P(C)  < 1- \sum_{i=0}^{k+1} \frac{\exp (-\lambda_k )[\lambda_k]^i}{i!}. \label{bound1}
\end{equation}
Given that $C$ occurs, the expected reward of Player A is bounded above by $t_{k+1}$ and the 
expected reward of Player B is bounded below by $v_k^k (t_{k+1})$, where from Equation
 (\ref{lower4})
\begin{equation}
v_k^k (t_{k+1}) = \frac{(k^2+k+1)t_{k+1}-t_{k+1}^2}{(k+1)^2}. \label{bound2}
\end{equation}
It follows that 
\begin{equation}
0<v^{\infty}_{-1} (e^{-1})-v^k_{-1} (e^{-1}) \leq \frac{kt_{k+1}+t_{k+1}^2}{(k+1)^2}
\left( 1- \sum_{i=0}^{k+1} \frac{\exp (-\lambda_k )[\lambda_k]^i}{i!}\right).
\label{bound3}
\end{equation}
Suppose that Player A accepts candidates on the interval $(t_0^{\infty},e^{-1})$ and 
Player B accepts candidates on the interval $(t_0^{k},e^{-1})$. It follows from the definition of these strategies that $t_0^k < t_0^{\infty}$. Solving Equation (\ref{lower8}), we obtain
that for $t\in (t_0^{\infty},e^{-1})$ 
\begin{equation}
v_{-1}^k (t) = -t\ln t + c_7 t; \hspace{.1in} v_{-1}^{\infty} (t) = -t\ln t + c_8 t.
\label{bound4}
\end{equation}
From the boundary conditions at $t=e^{-1}$, it follows that $c_8 >c_7$, since 
$v^{\infty}_{-1} (e^{-1})>v^k_{-1} (e^{-1})$. Hence, 
$v_{-1}^{\infty}(t)-v_{-1}^k (t)$ is increasing in $t$ on the interval $t\in (t_0^{\infty},e^{-1})$. 

Finally, $v^{\infty}_{-1}(t)$ is constant on the interval $(0,t_0^{\infty})$, while 
$v^k_{-1}(t)$ is decreasing on the interval $(t_0^k ,t_0^{\infty})$ and 
constant on the interval $(0,t_0^k)$. Thus
\begin{equation}
v^{\infty}_{-1}(0)\! -\! v^k_{-1}(0) \! \leq \! v^{\infty}_{-1}(e^{-1})\! -\! v^k_{-1}(e^{-1}) 
\! \leq \! \frac{kt_{k+1}+t_{k+1}^2}{(k+1)^2} \!
\left( 1 \! - \! \sum_{i=0}^{k+1} \frac{\exp (-\lambda_k )[\lambda_k]^i}{i!}\right). 
\label{bound5}
\end{equation}
The table below gives upper bounds on the fall in the value of the game to Player 2 caused by restricting his memory to $k$ for various $k$. The goal was to obtain a lower bound, 
$v_0^k (0)$, which is within approximately $10^{-6}$ of $v_{-1}^{\infty}(0)$. From 
Table \ref{tabbound} it appears that $k=3$ is sufficient for such accuracy. The appropriate lower bounds for the value of the game to Player 2 are also given. These were calculated 
by estimating $v^k (0,-1)$ (as illustrated in Example 
\ref{ex3}) with the aid of the Mathematica package (Wolfram, 2003) for $k=0,1,2,3$.
It follows that the value of the game to Player 2 lies in the interval $(v^3 (0,-1), u(0,-1))
\approx (0.199217,0.199548)$. 

\begin{table}[h]
\caption{Lower bounds on the value of the game to Player 2}\label{tabbound}%
\begin{tabular}{@{}lll@{}}
\toprule
$k$ & Upper bound on $v^{\infty}(0,-1)-v^k (0,-1)$ & Lower bound on $v(0,-1)$ [$v^k (0,-1)$] \\
\midrule
0 & $6.915\times 10^{-2}$ & 0.195684 \\
1 & $2.848\times 10^{-3}$ & 0.199086 \\
2 & $7.093\times 10^{-5}$ & 0.199214 \\
3 & $1.714\times 10^{-6}$ & 0.199217 \\
\botrule
\end{tabular}
\footnotetext{Source: Author's calculations with the aid of the Mathematica package}
\end{table}

\section{Conclusion} \label{sec6} 

This article has considered a variant of the best choice (secretary) problem that arises from 
a two-player game in which the players observe a sequence of objects asychronously and Player one always has priority. The problem faced by Player 2 is made difficult to solve by the fact that when Player 1 selects an object before Player 2 does, then Player 2 cannot 
compare the value of later objects with the object chosen by Player 1. Hence, there is 
uncertainty regarding the relative rank of such objects. Any such object that is the best seen so far by Player 2 (a candidate) must have a relative rank of at most two. In the best choice problem considered here, at any decision point Player 2 can infer whether an optimally behaving Player 1 has accepted an object or not. The 
probability that such a candidate has relative rank one (i.e. is the best of all the objects to appear so far) can be derived based on a form of Bayesian updating. These calculations are 
based on the appearance times of candidates (according to Player 2) that appear after the moment when an optimally behaving Player 1 is first prepared to accept an object of relative rank one. Each succesive candidate is more likely to have relative rank one than the previous candidate.

 It was shown that, after Player 1 has accepted an object, the optimal actions of Player 2 are based on a one-step look ahead rule. In this case, Player 2 
should accept a candidate if and only if the expected reward from accepting that candidate is greater than the expected reward from accepting the next candidate to appear. Under the optimal response, if Player 1 has already accepted an object, then Player 2 accepts a candidate at moment $n$ when the probability of such a candidate having relative rank 1 is 
at least $q_n$, where $q_n$ is decreasing in $n$. Unless $N$ is very large, such a strategy is relatively straightforward to implement step by step using a Bayesian approach.  However, since the probability of a candidate having relative rank one depends on the history of the process, it is difficult to describe the optimal strategy in \textit{a priori} form unless the number of objects in the sequence, $N$, is relatively small. The optimal strategy is described in \textit{a priori} form for the case $N=50$. A near-optimal strategy based on just the number of candidates that 
have appeared after Player 1 has accepted an object  (but not their appearance times) was discussed. This strategy is 
compared to the optimal strategy in the case $N=50$. 

Results were derived on the asymptotic value of the game to Player 2 when $N\rightarrow \infty$. An upper bound on this value can be derived by considering a problem in which Player 2 is given additional information on the relative rank of a candidate that has been rejected. A lower bound on this expected reward can be derived by considering a near optimal strategy based on the number of candidates to appear after Player 1 has selected an object. These bounds are relatively narrow. The optimal response of Player 2 gives an expected reward in the interval (0.199217, 0.199548). In addition, calculations indicate that only very 
minimal gains can be made by remembering more than three previous candidates.

One possible route for future research would be to consider best choice problems in which there are various forms of uncertainty regarding the relative rank of a current observation. Although Bayesian updating is relatively simple to implement in the problem considered here, it is possible that combinatorical approaches could also be employed to solve such problems.

Another possible route for future research would be to find better estimates of the asymptotic value of the game to Player 2. One possible approach to this problem would be to consider an analogous game in which objects appear according to a Poisson process. One might also consider strategies that take into account the appearance times of candidates after a proportion $e^{-1}$ of the objects have been observed, as well as the number of such candidates. 

Finally, future research might consider an analogous game with a larger number of players. However, the increase in the technical difficulty of finding a solution to the two-player game compared to the classical secretary problem suggests that solving the corresponding three-player game would be much more difficult, even if it suffices just to derive/estimate the optimal response of Player 3.

\bmhead{Acknowledgments}

The author thanks Prof. Micha\l \hspace{.02in} Morayne and Prof. Krzysztof Szajowski for help and conversions leading to this paper.

\bmhead{Funding} 

The author did not receive any financial funding in the writing of this manuscript.

\bmhead{Statement regarding Conflicts of Interest} 

The authors have no competing interests to declare that are relevant to the content of this article.

\bmhead{Statement regarding Data Availability} 

No datasets were generated or analyzed during the writing of this article.

\begin{appendices}

\section{ \hspace{-1.3em}: Proof of Theorem 1}\label{secA1}
 
Consider two states $(n,p)$ and $(n,r)$ where $n^{\ast}\leq n<N$ and $p>r\geq 0$. Assume that when starting from the state $(n,r)$ not more than $j$ future candidates should be rejected, where $j\geq 0$. It thus suffices to show by induction that starting in state $v(n,p)$ not more than $j$ future candidates should be rejected and $v(n,p)>v(n,r)$ for all $j\geq 0$. 

Suppose that $j=0$, i.e. starting from state $(n,r)$, the next candidate should always be accepted. It follows from the form of the optimal policy that it is always optimal to accept the next candidate when starting from state $(n,p)$. Hence, from Equation (\ref{opt4p2}) 
\begin{eqnarray*}
v(n,p)-v(n,r) & = & \frac{n(p-r)}{N} \sum_{k=n+1}^N \left[ \frac{1}{k-1}- \frac{n-1}{(k-1)(k-2)} \right] \\
& = &  \frac{n(p-r)}{N} \sum_{k=n+1}^N \frac{(k-2)-(n-1)}{(k-1)(k-2)} \geq 0.
\end{eqnarray*}

Now assume that the theorem holds for $j=i$. Assume that starting from state 
$(n,r)$ a maximum of $i+1$ candidates should be rejected. Also, when the next candidate to
appear is observed at moment $k$, then it should be accepted if and only if $k\geq k_0$. 
From Equation (\ref{opt5p2}) we have 
\begin{eqnarray}
v(n,r) & = & rn\left[\sum_{k=n+1}^{k_0-1} \frac{v(k,\tilde{q}[r;n,k])}{k(k-1)}\right] +2(1-r)n(n-1) \left[ \sum_{k=n+1}^{k_0-1} \frac{v(k,\tilde{q}[r;n,k])}{k(k-1)(k-2)} \right] + \ldots \nonumber \\
& & +\frac{rn}{N}\left[\sum_{k=k_0}^{N} \frac{1}{k-1}\right] +\frac{(1-r)n(n-1)}{N} \left[\sum_{k=k_0}^{N} \frac{1}{(k-1)(k-2)} \right] \label{ap1}, 
\end{eqnarray}
where $\tilde{q}[r;n,k]$ is the probability that the next candidate to appear has relative rank 1 given that it appears at time $k$ and the current candidate has relative rank 1 with probability $r$.  It should be noted that using Lemma \ref{lemma1} Equation (\ref{ap1}) can be written in the form 
\begin{eqnarray}
v(n,r) & = & rn\left[\sum_{k=n+1}^{k_0-1} \frac{v(k,\tilde{q}[r;n,k])}{k(k-1)}\right] +2(1-r)n(n-1) \left[ \sum_{k=n+1}^{k_0-1} \frac{v(k,\tilde{q}[r;n,k])}{k(k-1)(k-2)} \right] + \ldots \nonumber \\
& & + P(T> k_0-1)w(k_0-1,\tilde{s}[r;n,k_0]) \label{ap2}, 
\end{eqnarray}
where $T$ is the time at which the next candidate appears, $w(k_0-1,p)$ is the 
expected reward from accepting the next candidate to appear starting from the 
state $(k_0-1,p)$ and $\tilde{s}[r;n,k_0]$ is the probability that, when starting from the state 
$(n,r)$, the most recent candidate to appear has relative rank 1 given that no candidate appears between moment $n$ and moment $k_0$.

From Equation (\ref{induc1}), it follows that
\[
\tilde{q}[r;n,k] = \frac{r(k-2)+(1-r)(n-1)}{r(k-2)+2(1-r)(n-1)}.
\]
Note that from Equation (\ref{totprob}) the probability that the next candidate appears 
at time $k$ is given by $t(r;n,k)$, where
\[
t(r;n,k) = \frac{rn}{k(k-1)}+  \frac{2(1-r)n(n-1)}{k(k-1)(k-2)}.
\] 

Now consider the problem faced by Player 2 when starting from state $(n,p)$. We can obtain a lower bound on $v(n,p)$ by assuming that when the next candidate to
appear is observed at moment $k$, then it is accepted if and only if $k\geq k_0$. The probability that the next candidate appears 
at time $k$ is given by $t(p;n,k)$. Note that 
\[
t(r;n,k)-t(p;n,k) = \frac{(r-p)n}{k(k-1)}+  \frac{2(p-r)n(n-1)}{k(k-1)(k-2)} =
 \frac{n(p-r)[2(n-1)-(k-2)]}{k(k-1)(k-2)}. 
\]
Hence, for $k\leq 2n$, $t(r;n,k)-t(p;n,k) \geq 0$. 

Now suppose that $k>2n\geq 2n^{\ast}$. The probabilty that any candidate appearing at or after moment $2n^{\ast}$ has absolute rank 1 is $\geq \frac{1}{2}\times \frac{2n^{\ast}}{N}=\frac{n^{\ast}}{N}$. Since this is 
greater than the corresponding expected reward of a searcher in the standard secretary problem, Player 2 should always accept such a candidate.

Starting from state $(n,p)$, we assume that given a decision point has not occurred since moment $n$, then an artificial decision point occurs at moment $k$, $n<k<k_0$, with probability $t(r;n,k)-t(p;n,k)$. This is done to make the distribution of the next decision point given that a candidate is not accepted independent of the parameter $p$. Let $D$ be the event that such a decision point is a real decision point and $\tilde{q}(p;n,k)$ denote the probability that at moment $k$ the most recent candidate to appear has relative rank 1. Note that if the decision point at $k$ is real, then this is the 
probability that the current candidate has relative rank 1. If the decision point at $k$ is 
artificial, then the most recent candidate appeared at a previous time. However, the 
distribution of the time until the next candidate appears (i.e. decision point) only depends on the probability that 
the most recent candidate has relative rank one given the history of the process. For 
the purposes of this transformation of the problem, we may assume that the history of the 
process is given by the set of decision points (both real and artificial). Also, from the 
definition of an artifical decision point, Player 2 should never accept an object at such a moment. Let $T$ denote the moment at which a decision point (artificial or real) occurs. 
Starting from the state $(n,p)$, at the next decision point, the probability that the most recent candidate to appear has relative rank 1 is given by 
\[
\tilde{q}[p;n,k] \! = \! \frac{P (R_n \! = \! 1,D^c,T \! = \! k) \! + \! \sum_{i=1}^2 \! P(R_k \! = \! 1|D, R_n  \! = \! i, T  \! = \!  k) P(D, T  \! = \!  k, R_n \! = \! i)}{P(T= k)} 
\]
Note that $P(R_k =1|D,R_n =1, T=k)=1$, $P(R_k =1|D,R_n=2,T=k)=0.5$, $P(T=k)=t(r;n,k)$. In addition, 
\[
P(D,T=k,R_n =1) = \frac{pn}{k(k-1)}; \hspace{.2in} P(D,T=k,R_n =2) = \frac{2(1-p)n(n-1)}{k(k-1)(k-2)} .
\]
Thus
\begin{eqnarray*}
\tilde{q}[p;n,k] & > & \left(  \frac{pn}{k(k-1)}\times \frac{1}{t(r;n,k)}+\frac{2(1-p)n(n-1)}{k(k-1)(k-2)}\times \frac{1}{2t(r;n,k)} \right) \\
&  = & \frac{p(k-2)+(n-1)(1-p)}{r(k-2)+2(n-1)(1-r)}.
\end{eqnarray*}

Hence, for $k<k_0$
\[
\tilde{q}[p;n,k] -\tilde{q}[r;n,k] > \frac{(p-r)(k-n-1)}{r(k-2)+2(n-1)(1-r)} \geq 0. 
\]
Analogously to Equation (\ref{ap2}), we obtain 
\begin{eqnarray}
v(n,p) & > & rn\left[\sum_{k=n+1}^{k_0-1} \frac{v(k,\tilde{q}[p;n,k])}{k(k-1)}\right] +2(1-r)n(n-1) \left[ \sum_{k=n+1}^{k_0-1} \frac{v(k,\tilde{q}[p;n,k])}{k(k-1)(k-2)} \right] + \ldots \nonumber \\
& & + P(T> k_0-1)w(k_0-1,\tilde{s}[p;n,k_0]), \label{a3}
\end{eqnarray}

Note that when staring from state $(n,r)$, after rejecting a candidate at the next decision point $(k,\tilde{q}[r;n,k])$, where $k<k_0$, by assumption a maximum of $i$ candidates should be rejected. From the induction assumption, when staring from state $(n,p)$, after rejecting a possible candidate at the next (real or artificial) decision point $(k,\tilde{q}[p;n,k])$, where $k<k_0$, by assumption a maximum of $i$ candidates should be rejected, since 
$\tilde{q}[p;n,k]>\tilde{q}[r;n,k]$. It follows that starting from state $(n,p)$ a maximum of 
$i+1$ candidates should be rejected.

From the induction assumption it follows that the expressions in the two sums in Equation (\ref{a3}) are greater than the corresponding expressions in Equation (\ref{ap2}). Since 
$w(n,p)$ is increasing in $p$, it thus suffices to show that $\tilde{s}[p;n,k_0]>\tilde{s}[r;n,k_0]$.

Consider a process starting from the state $(n,r)$. Given that the most recent candidate 
has relative rank 1, from Equation (\ref{tm2}) the probability that no candidate appears before moment $k_0$ is given 
by $\frac{n}{k_0 -1}$. Given that the most recent candidate 
has relative rank 2, from Equation (\ref{tm9}) the probability that no candidate appears before moment $k_0$ is given 
by $\frac{n(n-1)}{(k_0 -1)(k_0-2)}$. Let $p\geq r$. According to the modified model, starting from state $(n,p)$ the 
probability that no decision point (artificial or real) occurs before moment $k_0$ is given by
$P(N)$, where
\[
P(N)=\frac{rn}{k_0-1} + \frac{(1-r)n(n-1)}{(k_0-1)(k_0-2)} .
\]
It should be noted that when $p>r$ the definition of the probability of an artificial decision point occurring at a given moment is not conditional on the relative rank of the most recent 
candidate. Given that the most recent candidate has relative rank 1, the probability that 
no decision point (artificial or real) occurs before moment $k_0$ is given by 
\[
P(N|R_n=1)=1-\sum_{k=n+1}^{k_0-1} \left( \frac{n}{k(k-1)} + \frac{n(p-r)[2(n-1)-(k-2)]}{k(k-1)(k-2)} \right).
\]
Note that here the first term in the sum is the probability that the $k$-th object is a candidate 
given that the relative rank of the previous candidate is 1 and the second term is the probability of an artificial decision point occurring at moment $k$. It follows that
\[
P(N|R_n=1)= \frac{(1-p+r)n}{k_0-1} + \frac{(p-r)n(n-1)}{(k_0-1)(k_0-2)}.
\]
Analogously, given that the most recent candidate has relative rank 2, the probability that 
no decision point (artificial or real) occurs before moment $k_0$ is given by 
\begin{eqnarray*}
P(N|R_n=2) & = & 1-\sum_{k=n+1}^{k_0-1} \left( \frac{2n(n-1)}{k(k-1)(k-2)} + \frac{n(p-r)[2(n-1)-(k-2)]}{k(k-1)(k-2)} \right) \\
& = &  \frac{(1+p-r)n(n-1)}{(k_0-1)(k_0-2)} -\frac{(p-r)n}{k_0-1}.
\end{eqnarray*}
Using Bayes' rule, the probability that the most recent candidate has relative rank 1 given that no candidate appears between moment $n$ and moment $k_0$ is given by 
\[
\tilde{s}[p;n,k_0] \! = \! P(R_n \! =\! 1|N) \!  = \! \frac{P(N|R_n \! =\! 1)P(R_n \! = \! 1)}{P(N)} \! = \! \frac{p[(1\! -\! p\! +\! r)(k_0\! -\! 2)\! +\! (p\! -\! r)(n\! -\! 1)]}{r(k_0-2)+(1-r)(n-1)}.
\]
Differentiating this with respect to $p$, we obtain
\[
\frac{d\tilde{s}[p;n,k_0]}{dp} = \frac{k_0-2+(2p-r)[(n-1)-(k_0 -2)]}{r(k_0 -2)+(1-r)(n-1)}.
\]
The denominator of this expression is positive and independent of $p$. The numerator is 
a linear function of $p$ and is defined for $r\leq p\leq 1$. When $p=r$, the numerator is 
given by 
\[
(1-p)(k_0-2)+p(n-1)>0.
\]
When $p=1$, the numerator is 
given by 
\[
2(n-1)-(k_0-2)+r[(k_0-2)-(n-1)].
\]
Note that this expression is bounded below by 0, since $n-1\leq k_0-2 \leq 2(n-1)$. It follows that $\tilde{s}[p;n,k_0] $ is increasing in $p$.
 It thus follows from the induction hypothesis that $v(n,p)>v(n,r)$ for $p>r\geq 0$.

Finally, we prove that for $n^{\ast}\leq n<N$, $v(n,-1)< v(n,0)$. This is proven using 
recursion on $n$. Note that $v(N-1,-1)=0<v(N-1,0)=\frac{1}{N}$. Now assume that 
$v(n+1,-1)<v(n+1,0)$. Conditioning on whether the $(n+1)$-th object has relative rank 1 or not
\[
v(n,-1)= \frac{nv(n+1,-1)}{n+1}+\frac{v(n+1,0)}{n+1}.
\]
Similarly, conditioning on whether the $(n+1)$-th object has relative rank at most 2 or not
\[
v(n,0)=\frac{(n-1)v(n+1,0)}{n+1}+\frac{2}{n+1} \max\left\{ v(n+1,0.5),\frac{n+1}{2N}\right\} 
\geq v(n+1,0).
\] 
It follows that 
\begin{eqnarray*}
v(n,0)-v(n,-1) & \geq & v(n+1,0)-\frac{nv(n+1,-1)}{n+1}-\frac{v(n+1,0)}{n+1} \\
& \geq & \frac{n}{n+1} [v(n+1,0)-v(n+1,-1)] >0.
\end{eqnarray*}
The result then follows by recursion.

\end{appendices}

\end{document}